\documentclass[prodmode,acmtecs]{acmsmall}
\usepackage{anysize}
\marginsize{3cm}{3cm}{4cm}{2.5cm}
\usepackage[boxed,ruled,vlined,linesnumbered]{algorithm2e}
\usepackage{framed}
\usepackage{rotating,booktabs}
\usepackage{multirow}
\usepackage{hhline}
\usepackage{lscape}

\usepackage[T1]{fontenc}
\usepackage[utf8]{inputenc}
\usepackage{tabularx,ragged2e,booktabs,caption}
\newcolumntype{C}[1]{>{\Centering}m{#1}}

\usepackage{rotating}
\usepackage{csquotes}
\usepackage{enumitem}
\usepackage{lscape}
\usepackage{rotating}
\usepackage{multirow}
\usepackage{url}
\newcommand{\B}{\vspace*{-\smallskipamount}}

\newtheorem{claim}[theorem]{Claim}

\usepackage[none]{hyphenat}
\usepackage{amsmath}

\usepackage{lineno}
\usepackage{wrapfig}
\usepackage{amssymb}
\usepackage{xspace}
\usepackage{tikz}
\usetikzlibrary{automata,matrix,shapes,arrows,positioning,chains,calc}
\usetikzlibrary{snakes}
\usetikzlibrary{arrows,scopes}
\usetikzlibrary{positioning,chains,fit,shapes,calc}

\makeatletter
\def\@copyrightspace{\relax}
\makeatother

\begin{document}

\markboth{F. Afrati et al.}{Assignment Problems of Different-Sized Inputs in MapReduce}
\title{Assignment Problems of Different-Sized Inputs in MapReduce}
\author{FOTO AFRATI
\affil{National Technical University of Athens, Greece}
SHLOMI DOLEV, EPHRAIM KORACH, and SHANTANU SHARMA
\affil{Ben-Gurion University, Israel}
JEFFREY D. ULLMAN
\affil{Stanford University, USA}}

\begin{abstract}
A MapReduce algorithm can be described by a \emph{mapping schema}, which assigns inputs to a set of reducers, such that for each required output there exists a reducer that receives all the inputs that participate in the computation of this output. Reducers have a capacity, which limits the sets of inputs that they can be assigned. However, individual inputs may vary in terms of size. We consider, for the first time, mapping schemas where input sizes are part of the considerations and restrictions. One of the significant parameters to optimize in any MapReduce job is communication cost between the map and reduce phases. The communication cost can be optimized by minimizing the number of copies of inputs sent to the reducers. The communication cost is closely related to the number of reducers of constrained capacity that are used to accommodate appropriately the inputs, so that the requirement of how the inputs must meet in a reducer is satisfied. In this work, we consider a family of problems where it is required that each input meets with each other input in at least one reducer. We also consider a slightly different family of problems in which, each input of a list, $X$, is required to meet each input of another list, $Y$, in at least one reducer. We prove that finding an optimal mapping schema for these families of problems is NP-hard, and present a bin-packing-based approximation algorithm for finding a near optimal mapping schema.
\end{abstract}

\category{H.2.4}{Systems}{Parallel Databases} \category{H.2.4}{Systems}{Distributed Databases} \category{C.2.4}{Distributed Systems}{Distributed Databases}

\terms{Design, Algorithms, Performance}

\keywords{Distributed computing, mapping schema, MapReduce algorithms, reducer capacity, and reducer capacity and communication cost tradeoff}

\begin{bottomstuff}
This paper is accepted in \textbf{ACM Transactions on Knowledge Discovery from Data (TKDD), August 2016.} Preliminary versions of this paper have appeared in the proceeding of DISC 2014 and BeyondMR 2015~\cite{A2A}.

This work of F. Afrati is supported by the project Handling Uncertainty in Data Intensive Applications, co-financed by the European Union (European Social Fund) and Greek national funds, through the Operational Program ``Education and Lifelong Learning,'' under the program THALES. This work of S. Dolev is partially supported by Rita Altura Trust Chair in Computer Sciences, Lynne and William Frankel Center for Computer Sciences, Israel Science Foundation (grant number 428/11), Cabarnit Cyber Security MAGNET Consortium, and Ministry of Science and Technology, Infrastructure Research in the Field of Advanced Computing and Cyber Security.

Author's addresses: F. Afrati, School of Electrical and Computing Engineering, National Technical University of Athens, Greece (e-mail: \texttt{afrati@softlab.ece.ntua.gr}), S. Dolev, Department of Computer Science, Ben-Gurion University of the Negev, Beer-Sheva, Israel (e-mail: \texttt{dolev@cs.bgu.ac.il}), E. Korach, Department of Industrial Engineering and Management, Ben-Gurion University of the Negev, Israel (e-mail: \texttt{korach@bgu.ac.il}), S. Sharma, Department of Computer Science, Ben-Gurion University of the Negev, Beer-Sheva, Israel (e-mail: \texttt{sharmas@cs.bgu.ac.il}), J.D. Ullman, Department of Computer Science, Stanford University, USA (e-mail: \texttt{ullman@cs.stanford.edu}).
\end{bottomstuff}

\maketitle

\section{Introduction}
\label{sec:introduction}
MapReduce~\cite{DBLP:conf/osdi/DeanG04} is a programming system used for parallel processing of large-scale data. It has two phases, the {\em map phase} and the {\em reduce phase}. The given input data is processed by the map phase that applies a user-defined map function to produce intermediate data (of the form $\langle key, value \rangle$). Intermediate data is, then, processed by the reduce phase that applies a user-defined reduce function to keys and their associated values. The final output is provided by the reduce phase. A detailed description of MapReduce can be found in Chapter 2 of~\cite{DBLP:books/ullman2011}.

\medskip\medskip\noindent \textbf{Communication Cost and Reducer Capacity.} An important performance measure for MapReduce algorithms is the amount of data transferred from the \emph{mappers} (the processes that implement the map function) to the \emph{reducers} (the processes that implement the reduce function). This is called the \emph{communication cost}. The minimum communication cost is, of course, the size of the desired inputs that provide the final output, since we need to transfer all these inputs from the mappers to the reducers at least once. However, we may need to transfer the same input to several reducers, thus increasing the communication cost.

Depending on various factors of our setting, each reducer may process a larger or smaller amount of data. The amount of data each reducer processes however affects the wall clock time of our algorithms and the degree of parallelization. If we send all data in one reducer, then we have low communication (equal to the size of the data) but we have low degree of parallelization, and thus the wall clock time increases. Thus, the maximum amount of data a reducer can hold is a constraint when we build our algorithm.

\noindent\emph{Reducer capacity.} We define \emph{reducer capacity} to be the upper bound on the sum of the sizes of the $value$s that are assigned to the reducer. For example, we may choose the reducer capacity to be the size of the main memory of the processor on which the reducer runs or we may arbitrarily set a low reducer capacity if we want high parallelization. We always assume in this paper that all the reducers have an identical capacity, denoted by $q$.

There are various works in the field of MapReduce algorithms design (\textit{e}.\textit{g}.,~\cite{DBLP:conf/soda/KarloffSV10,DBLP:journals/crossroads/Ullman12,DBLP:journals/pvldb/AfratiSSU13,DBLP:journals/corr/abs-1004-4708,DBLP:conf/ics/PietracaprinaPRSU12,DBLP:conf/ideas/AfratiU13}) that investigate problems and/or build algorithms with minimum communication cost when the reducer size is bounded by the number of inputs that a reducer is allowed to hold. In this paper, we consider for the first time problems where each input may have a different size and the {\em reducer capacity}
is an upper bound on the sum of the sizes of the inputs in a reducer. Here, we investigate the problem where each input is required to meet in a reducer with any other input. We give now some examples where this problem may appear in practice.

\medskip\medskip
\noindent \textbf{Motivating Examples.} We present three examples.
\begin{example}
\noindent \textit{Computing common friends.} An input is a list of friends. We have such lists for $m$ persons. Each pair of lists of friends corresponds to one output, which will show us the common friends of the respective persons. Thus, it is mandatory that lists of friends of every two persons are compared. Specifically, the problem is: a list $F=\{f_1, f_2, \ldots, f_m\}$ of $m$ friends is given, and each pair of elements $\langle f_i,f_j\rangle$ corresponds to one output, common friends of persons $i$ and $j$; see Figure~\ref{fig:friend}.
\end{example}

\begin{figure}[h]
    \begin{minipage}[t]{0.45\linewidth}
        \centering
          \includegraphics[width=55mm, height=40mm]{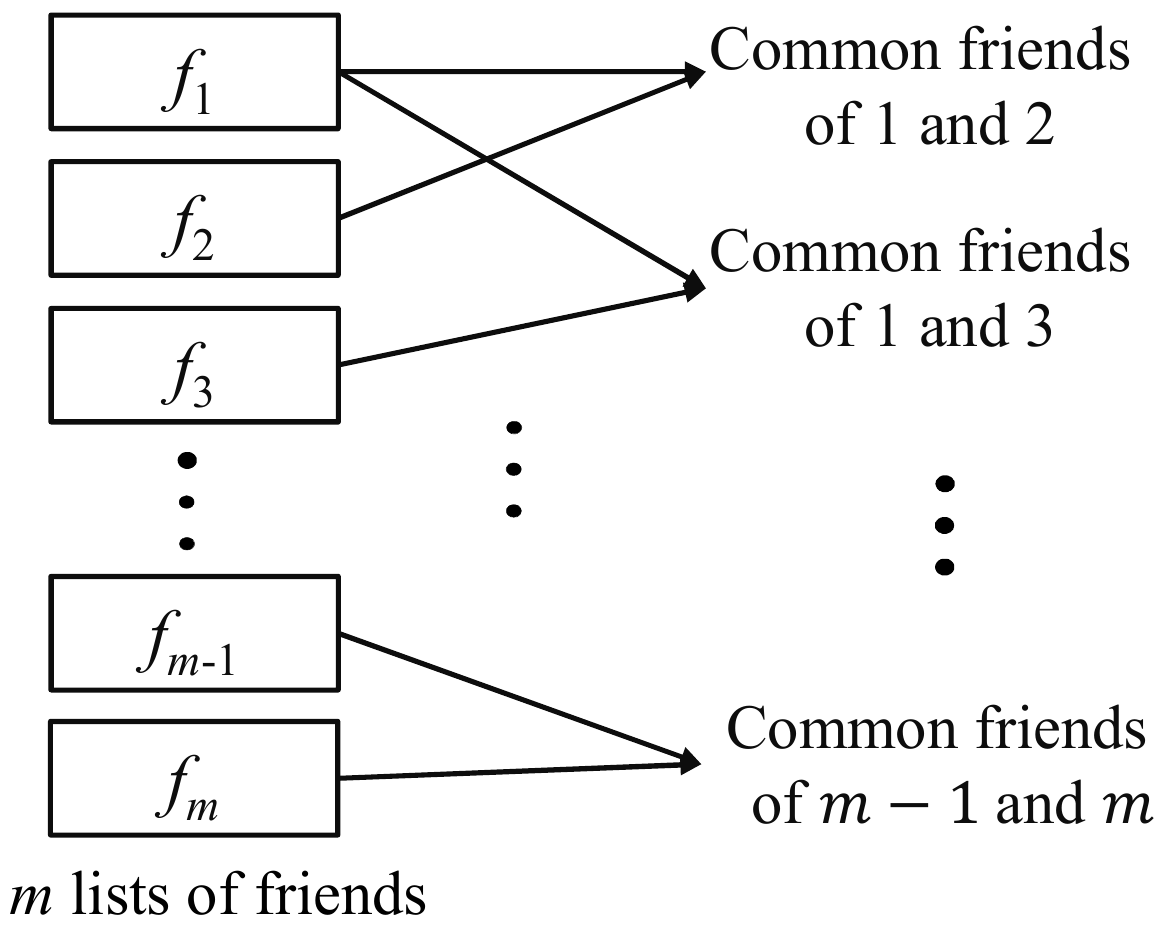}
      \caption{Computing common friends example.}
       \label{fig:friend}
    \end{minipage}
    \quad\quad
    \begin{minipage}[t]{0.49\linewidth}
        \centering
           \includegraphics[width=65mm, height=35mm]{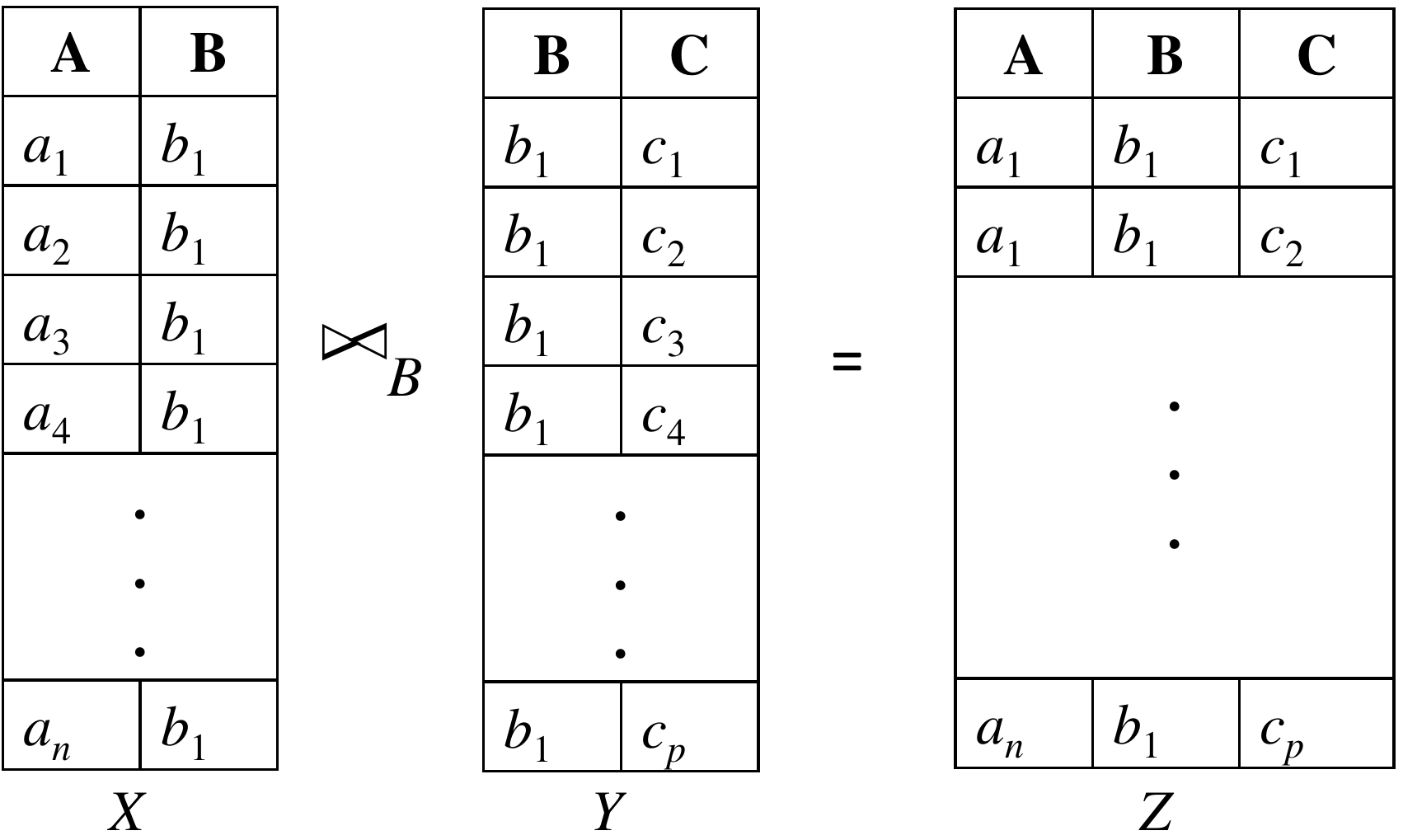}
            \caption{Skew join example for a heavy hitter, $b_1$.}
        \label{fig:twowayjoin}
    \end{minipage}
\end{figure}

\begin{example}
\textit{Similarity-join.} Similarity-join is an example of the \emph{A2A mapping schema problem} that can be used to find the similarity between any two inputs, \textit{e}.\textit{g}., Web pages or documents. A set of $m$ inputs (\textit{e}.\textit{g}., Web pages) $\mathit{WP}= \{wp_1, wp_2, \ldots, wp_m\}$, a similarity function $sim(x, y)$, and a similarity threshold $t$ are given, and each pair of inputs $\langle wp_x, wp_y\rangle$ corresponds to one output such that $sim(wp_x, wp_y) \geq t$.

It is necessary to compare all-pairs of inputs when the similarity measure is sufficiently complex that shortcuts like locality-sensitive hashing are not available. Therefore, it is mandatory that every two inputs (Web pages) of the given input set ($\mathit{WP}$) are compared. The similarity-join is useful in various applications, mentioned in~\cite{Bayardo:2007:SUP:1242572.1242591}, \emph{e}.\emph{g}., near-duplicate document detection, collaborative filtering, and query refinement for Web search.
\end{example}

\begin{example}
\noindent \textit{The drug-interaction problem.} The drug-interaction problem is given in~\cite{DBLP:journals/crossroads/Ullman12}, where a list of inputs consists of 6,500 drugs and a drug $i$ holds information about the medical history of patients who had taken the drug $i$. The objective is to find pairs of drugs that had particular side effects. In order to achieve the objective, it is mandatory that each pair of drugs is compared.
\end{example}

\begin{example}
\noindent \textit{Skew join of two relations $X(A,B)$ and $Y(B,C)$.} The join of relations $X(A,B)$ and $Y(B,C)$, where the joining attribute is $B$, provides output tuples $\langle a, b, c\rangle$, where $(a,b)$ is in $A$ and $(b,c)$ is in $C$. One or both of the relations $X$ and $Y$ may have a large number of tuples with an identical $B$-value. A value of the joining attribute $B$ that occurs many times is known as a {\em heavy hitter}. In skew join of $X(A,B)$ and $Y(B,C)$, all the tuples of both the relations with an identical heavy hitter should appear together to provide the output tuples.

In Figure~\ref{fig:twowayjoin}, $b_1$ is considered as a heavy hitter; hence, it is required that all the tuples of $X(A,B)$ and $Y(B,C)$ with the heavy hitter, $B= b_1$, should appear together to provide the desired output tuples, $\langle a, b_1, c\rangle$ ($a\in A, b_1\in B, c\in C$), which depend on exactly two inputs. \end{example}

\medskip\medskip
\noindent \textbf{Problem Statements.}
We define two problems where exactly two inputs are required for computing an output:
\begin{description}
\item[\emph{All-to-All problem}.] In the \emph{all-to-all} (\textit{A2A}) problem, a list of inputs is given, and each pair of inputs corresponds to one output.

  \item [\emph{X-to-Y problem}.] In the \emph{X-to-Y} (\textit{X2Y}) problem, two disjoint lists $X$ and $Y$ are given, and each pair of elements $\langle x_i, y_j\rangle$, where $x_i \in X, y_j \in Y, \forall i, j$, of the lists $X$ and $Y$ corresponds to one output.
\end{description}
Computing common friends on a social networking site, and the drug-interaction problem are examples of A2A problems.
Skew join is an example of a X2Y problem.

A {\em mapping schema} defines a MapReduce algorithm. A {\em mapping schema} assigns input to reducers, so that no reducer exceeds the reducer capacity and all pairs of inputs (in A2A problem) or all pairs of X-to-Y inputs (in X2Y problem) meet in the same reducer.\footnote{For more general problems, we are given the graph which defines which pairs of inputs should meet in the same reducer to solve the problem and this is what the mapping schema should achieve -- but
we do not consider such problems here.}

The {communication cost, is a significant factor in the performance of a MapReduce algorithm. The communication cost comes with a tradeoff in the degree of parallelism, as we mentioned. A mapping schema is {\em optimal} if there is no other mapping schema with a lower communication cost. In this paper, we investigate how to construct optimal mapping schemas or good approximations of them.

\medskip\medskip \noindent \textbf{Outline of Paper and Our Contribution.}
In this paper, we investigate the problem of finding an optimal or near optimal mapping schema for the case we have inputs of different sizes.
\begin{itemize}
\item In Section~\ref{section:the_system_setting}, we warm up to the problem with discussing how the tradeoffs appear.

\item In Section~\ref{sec:Intractability of Finding a Mapping Schema}, we prove that finding an optimal mapping schema is intractable.

\item In Section~\ref{sec:Approximation Algorithms Preliminary Results}, we present preliminary results and present one of our techniques to obtain near optimal mapping schemas. The technique is to do bin-packing first and collect inputs in bins, then treat bins as inputs, possibly all of equal size.

\item In Section~\ref{sec:Equal-Sized Inputs Optimal Algorithms}, we present algorithms to construct optimal mapping schemas in certain cases where the inputs are all of equal size.

\item In Sections~\ref{subsec:Generalizing the Technique from section q is equal to 3} and~\ref{subsec:A hybrid approach}, we combine bin-packing and algorithmic techniques from Section~\ref{sec:Equal-Sized Inputs Optimal Algorithms} to build algorithms that construct mapping schemas that are good approximations to the optimal. For each algorithm, we argue in the end how good an approximation this is.

\item In Section~\ref{subsec:Generalizing Techniques from the AU method}, we extend the idea presented in Section~\ref{subsec:AU method} for equal size inputs.

\item In Sections~\ref{subsec:Generalizing the Technique from section q is equal to 3} and~\ref{subsec:A hybrid approach}, we only considered the case when there is no input of size $>\frac{q}{2}$ (remember we denote with $q$ the reducer capacity). Thus in Section~\ref{sec:A big input of size greater than q2}, we investigate the case where there is an input of size $>\frac{q}{2}$. We mainly use similar techniques as in Section~\ref{sec:Approximation Algorithms Preliminary Results}.

\item So far we have investigated the A2A problem. In Section~\ref{sec:a heuristic for the X-meets-Y Mapping Schema Problem}, we take the X2Y problem to provide algorithms for this too.
\end{itemize}

\medskip\medskip\noindent \textbf{Related Work.}
MapReduce was introduced by Dean and Ghemawat in 2004~\cite{DBLP:conf/osdi/DeanG04}. Karloff et al.~\cite{DBLP:conf/soda/KarloffSV10} presents a model for comparing MapReduce with the Parallel Random Access Machine (PRAM) model and states that a large class of PRAM algorithms can be simulated by MapReduce. However, parallel and sequential computations (used in MapReduce) differentiate MapReduce and PRAM model. Another model considers the efficiency of MapReduce algorithms in terms of algorithm's running time, suggested in~\cite{DBLP:journals/corr/abs-1004-4708}. The author simulates PRAM algorithms by MapReduce and defines memory-bound for MapReduce algorithms in terms of reducer I/O sizes for each round and each reducer.

Following~\cite{DBLP:conf/soda/KarloffSV10,DBLP:journals/corr/abs-1004-4708}, a filtering technique for MapReduce is suggested in~\cite{DBLP:conf/spaa/LattanziMSV11}. This technique removes some of nonessential data and results in fewer rounds than in both the previous stated models~\cite{DBLP:conf/soda/KarloffSV10,DBLP:journals/corr/abs-1004-4708}. Essentially, the models, in~\cite{DBLP:conf/soda/KarloffSV10,DBLP:journals/corr/abs-1004-4708,DBLP:conf/spaa/LattanziMSV11}, provide a way to simulate a large family of PRAM algorithms by MapReduce.

Afrati et al.~\cite{DBLP:journals/pvldb/AfratiSSU13} presents a model for MapReduce algorithms where an output depends on two inputs, and shows a tradeoff between the communication cost and parallelism. In \cite{DBLP:conf/ideas/AfratiU13}, the authors consider a case where each pair of inputs produces an output and present an upper bound that meets the lower bound on the communication cost as a function of the number of inputs sent to a reducer. However, both in \cite{DBLP:journals/pvldb/AfratiSSU13} and \cite{DBLP:conf/ideas/AfratiU13} the authors regard the reducer capacity in terms of the number of inputs (assuming each input is of an identical size) sent to a reducer.

Our setting is closely related to the settings given by Afrati et al.~\cite{DBLP:journals/pvldb/AfratiSSU13}, but we allow the input sizes to be different. To the best of our knowledge, we for the first time do not restrict the input sizes to be identical. Thus, we consider a more realistic settings for MapReduce algorithms that can be used in various practical scenarios.

\medskip
\section{Mapping Schema and Tradeoffs}
\label{section:the_system_setting}
Our system setting is an extension of the standard system setting~\cite{DBLP:journals/pvldb/AfratiSSU13} for MapReduce algorithms, where we consider, for the first time, inputs of different sizes. In this section, we provide formal definitions and some examples to show the tradeoff between communication cost
and degree of parallelization.

\medskip \noindent \textbf{Mapping Schema.} A mapping schema is an assignment of the set of inputs to some given reducers so that the following two constraints are satisfied:

\begin{itemize}
  \item A reducer is assigned inputs whose sum of the sizes is less than or equal to the reducer capacity $q$.
  \item For each output, we must assign its corresponding inputs to at least one reducer in common.
\end{itemize}
A mapping schema is optimal when the communication cost is minimum. The number of reducers we use often is minimal for an optimal mapping schema but this may not always be the case. It is desirable to minimize the number of reducers too. We offer insight about communication cost and number of reducers uses in Examples~\ref{example:a2a example} and~\ref{example:x2y example}.

\medskip \noindent \textbf{Tradeoffs.}
The following tradeoffs appear in MapReduce algorithms and in particular in our setting:
\begin{itemize}
  \item A tradeoff between the reducer capacity and the number of reducers. For example, large reducer capacity allows the use of a smaller number of reducers.
  \item A tradeoff between the reducer capacity and the parallelism. For example, if we want to achieve a high degree of parallelism, we set low reducer capacity.
  \item A tradeoff between the reducer capacity and the communication cost. For example, in the case reducer capacity is equal to the total size of the data then we can use one reducer and have minimum communication (of course, this goes at the expense of parallelization).
\end{itemize}
In the subsequent subsections, we present the A2A mapping schema problem and the X2Y mapping schema problem with fitting examples and explain the tradeoffs.

\medskip
\subsection{The \emph{A2A Mapping Schema Problem}}
\label{subsec:All-to-All Mapping Schema Problem_system_setting}

An instance of the \emph{A2A mapping schema problem} consists of a list of $m$ inputs whose input size list is $W=\{w_1, w_2, \ldots, w_m\}$ and a set of $z$ identical reducers of capacity $q$. A solution to the \emph{A2A mapping schema problem} assigns every pair of inputs to at least one reducer in common, without exceeding $q$ at any reducer.

\begin{figure}[!h]
 \centering
 \includegraphics{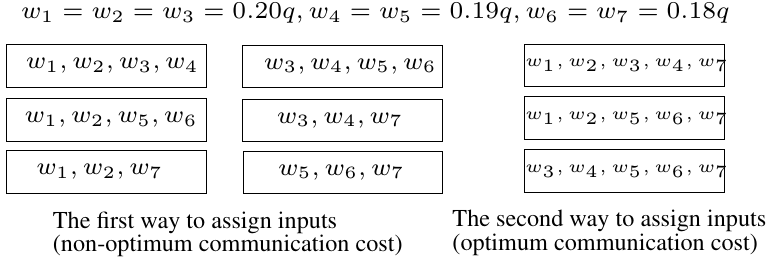}
 \caption{An example to the \emph{A2A mapping schema problem}.}
 \label{fig:A2A_assignment_example}
\end{figure}

\begin{example}\label{example:a2a example}
\noindent We are given a list of seven inputs $I=\{i_1, i_2, \ldots, i_7\}$ whose size list is $W=\{0.20q, 0.20q, 0.20q, 0.19q, 0.19q, 0.18q, 0.18q\}$ and reducers of capacity $q$. In Figure~\ref{fig:A2A_assignment_example}, we show two different ways that we can assign the inputs to reducers. The best we can do to minimize the communication cost is to use three reducers. However, there is less parallelism at the reduce phase as compared to when we use six reducers. Observe that when we use six reducers, then all reducers have a lighter load, since each reducer may have capacity less than $0.8q$.

The communication cost for the second case (3 reducers) is approximately $3q$, whereas for the first case
(6 reducers) it is approximately $4.2q$. Thus, in tradeoff, in the 3-reducers case we have low communication cost but also lower degree of parallelization, whereas in the 6-reducers case we have
high parallelization at the expense of the communication cost.

\end{example}

\begin{figure}[!h]
 \centering
 \includegraphics[scale=0.4]{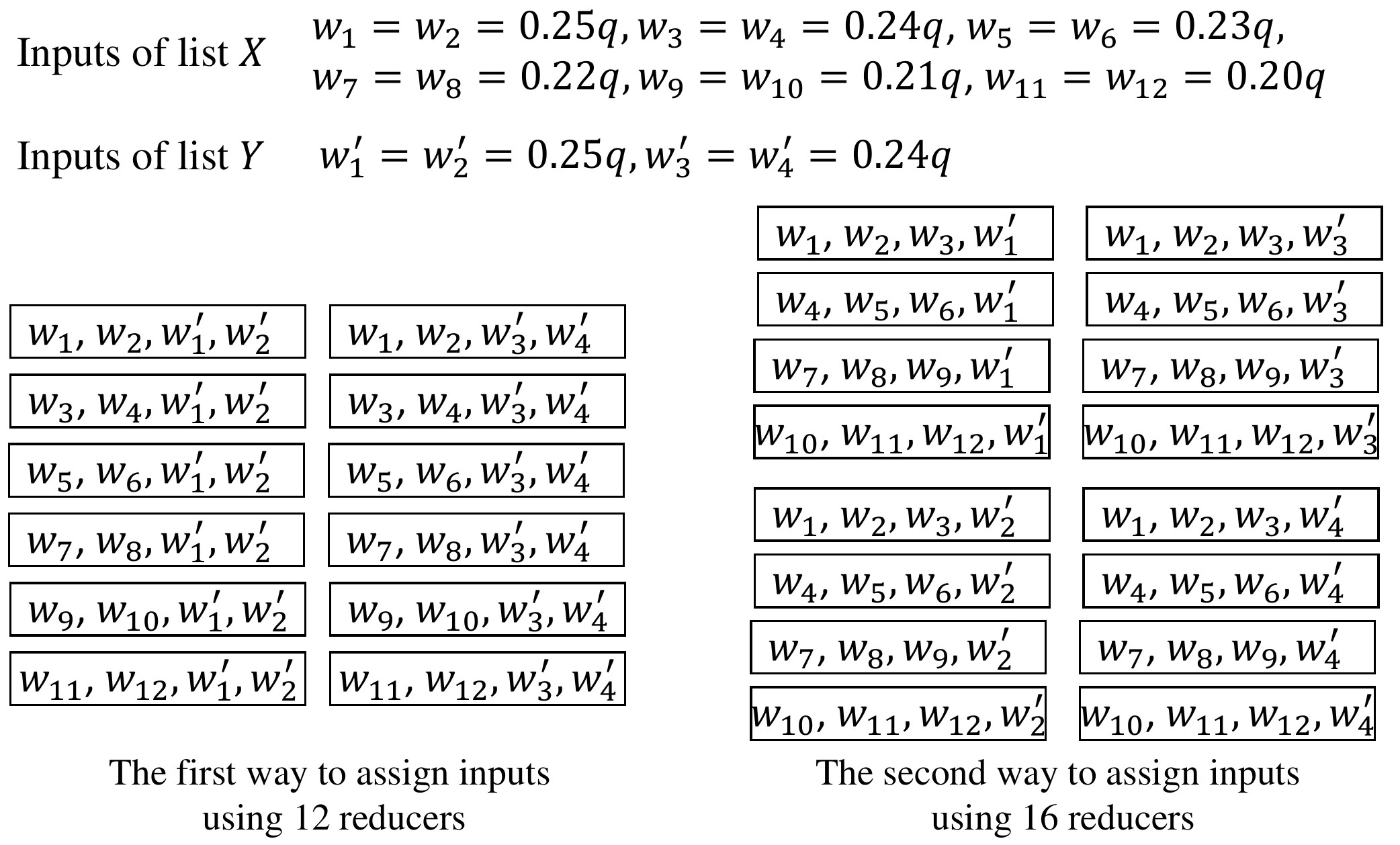}
 \caption{An example to the \emph{X2Y mapping schema problem}.}
 \label{fig:X2Y_assignment_example}
\end{figure}

\medskip
\subsection{The \emph{X2Y Mapping Schema Problem}}
\label{subsec:X-meets-Y mapping schema problem_system_setting}
An instance of the \emph{X2Y mapping schema problem} consists of two disjoint lists $X$ and $Y$ and a set of identical reducers of capacity $q$. The inputs of the list $X$ are of sizes $w_1, w_2, \ldots, w_m$, and the inputs of the list $Y$ are of sizes $w_1^{\prime}, w_2^{\prime}, \ldots, w_n^{\prime}$. A solution to the \emph{X2Y mapping schema problem} assigns every two inputs, the first from one list, $X$, and the second from the other list, $Y$, to at least one reducer in common, without exceeding $q$ at any reducer.

\begin{example}\label{example:x2y example}
\noindent We are given two lists, $X$ of 12 inputs, and $Y$ of 4 inputs (see Figure~\ref{fig:X2Y_assignment_example}) and reducers of capacity $q$. We show that we can assign each input of the list $X$ with each input of the list $Y$ in two ways. In order to minimize the communication cost, the best way is to use 12 reducers. Note that we cannot obtain a solution for the given inputs using less than 12 reducers. However, the use of 12 reducers results in less parallelism at the reduce phase as compared to when we use 16 reducers.
\end{example}

\begin{itemize}
\item In this paper, we assume we have made a decision on the degree of parallelization we want (by setting the reducer capacity $q$).
\end{itemize}

\medskip
\section{Intractability of Finding a Mapping Schema}
\label{sec:Intractability of Finding a Mapping Schema}
In this section, we will show that the \emph{A2A} and the \emph{X2Y} mapping schema problems do not possess a polynomial solution. In other words, we will show that the assignment of \textit{two required inputs} to the minimum number of identical-capacity reducers to find solutions to the \emph{A2A} and the \emph{X2Y} mapping schema problems cannot be achieved in polynomial time.

\medskip
\subsection{NP-hardness of the \emph{A2A Mapping Schema Problem}}
\label{subsection:the_all-to-all_mapping_schema_problem}

A list of inputs $I=\{i_1,i_2,\ldots,i_m\}$ whose input size list is $W=\{w_1, w_2, \ldots, w_m\}$ and a set of identical reducers $R=\{r_1,r_2,\ldots,r_z\}$, are an input instance to the \emph{A2A mapping schema problem}. The \emph{A2A mapping schema problem} is a decision problem that asks whether or not there exists a mapping schema for the given input instance such that every input, $i_x$, is assigned with every other input, $i_y$, to at least one reducer in common. An answer to the \emph{A2A mapping schema problem} will be \enquote{yes,}
if for each pair of inputs ($\langle i_x,i_y\rangle$), there is at least one reducer that holds them.

In this section, we prove that the \emph{A2A mapping schema problem} is NP-hard in the case of $z>2$ identical reducers. In addition, we prove that the \emph{A2A mapping schema problem} has a polynomial solution to one and two reducers.

If there is only one reducer, then the answer is \enquote{yes} if and only if the sum of the input sizes $\sum_{i=1}^m w_i$ is at most $q$. On the other hand, if $q< \sum_{i=1}^m w_i$, then the answer is \enquote{no.} In case of two reducers, if a single reducer is not able to accommodate all the given inputs, then there must be at least one input that is assigned to only one of the reducers, and hence, this input is not paired with all the other inputs. In that case, the answer is \enquote{no.} Therefore, we achieve a polynomial solution to the \emph{A2A mapping schema problem} for one and two identical-capacity reducers.

We now consider the case of $z>2$ and prove that the \emph{A2A mapping schema problem} for $z>2$ reducers is at least as hard as the partition problem.

\medskip\begin{theorem}
\label{th:all-to-all}
The problem of finding whether a mapping schema of $m$ inputs of different input sizes exists, where every two inputs are assigned to at least one of $z\geq 3$ identical-capacity reducers, is NP-hard.
\end{theorem}

The proof appears in Appendix~\ref{app:proof of th}.

\medskip
\subsection{NP-hardness of the \emph{X2Y Mapping Schema Problem}}
\label{subsection:the_x-meets-y_mapping_schema_problem}
Two lists of inputs, $X=\{i_1,i_2,\ldots,i_m\}$ whose input size list is $W_x=\{w_1, w_2, \ldots, w_m\}$ and $Y=\{i_1^{\prime}, i_2^{\prime}, \ldots, i_n^{\prime}\}$ whose input size list is $W_y=\{w_1^{\prime}, w_2^{\prime}, \ldots, w_n^{\prime}\}$, and a set of identical reducers $R=\{r_1,r_2,\ldots,r_z\}$ are an input instance to the \emph{X2Y mapping schema problem}. The \emph{X2Y mapping schema problem} is a decision problem that asks whether or not there exists a mapping schema for the given input instance such that each input of the list $X$ is assigned with each input of the list $Y$ to at least one reducer in common. An answer to the \emph{X2Y mapping schema problem} will be \enquote{yes,} if for each pair of inputs, the first from $X$ and the second from $Y$, there is at least one reducer that has both those inputs.

The \emph{X2Y mapping schema problem} has a polynomial solution for the case of a single reducer. If there is only one reducer, then the answer is \enquote{yes} if and only if the sum of the input sizes $\sum_{i=1}^m w_i+\sum_{i=1}^n w_i^{\prime}$ is at most $q$. On the other hand, if $q< \sum_{i=1}^m w_i+\sum_{i=1}^n w_i^{\prime}$, then the answer is \enquote{no.} Next, we will prove that the \emph{X2Y mapping schema problem} is an NP-hard problem for $z>1$ identical reducers.

\medskip\begin{theorem}
\label{th:x-meets-y}
The problem of finding whether a mapping schema of $m$ and $n$ inputs of different input sizes that belongs to list $X$ and list $Y$, respectively, exists, where every two inputs, the first from $X$ and the second from $Y$, are assigned to at least one of $z\geq2$ identical-capacity reducers, is NP-hard.
\end{theorem}

The proof appears in Appendix~\ref{app:proof of th}.

\medskip
\section{Approximation Algorithms: Preliminary Results}
\label{sec:Approximation Algorithms Preliminary Results}
Since the \emph{A2A Mapping Schema Problem} is NP-hard, we start looking at special cases and developing approximation algorithm to solve it. We propose several approximation algorithms for the \emph{A2A mapping schema problem} that are based on bin-packing algorithms, selection of a prime number $p$, and division of inputs into two sets based on their sizes.

Each algorithm takes the number of inputs, their sizes, and the reducer capacity (see Table~\ref{table:Reducer capacity and input constraints for different algorithms for the mapping schema problems}). The approximation algorithms have two cases depending on the sizes of the inputs, as follows:
\begin{enumerate}
  \item Input sizes are upper bounded by $\frac{q}{2}$.
  \item One input is of size, say $w_i$, greater than $\frac{q}{2}$, but less than $q$, and all the other inputs have size less than or equal to $q- w_i$. In this case most of the communication cost comes from having to pair the large input with every other input.
\end{enumerate}

Of course, if the two largest inputs are greater than the given reducer capacity $q$, then there is no solution to the \emph{A2A mapping schema problem} because these two inputs cannot be assigned to a single reducer in common.

\medskip\medskip\noindent\textbf{Parameters for analysis.} We analyze our approximation algorithms on the following parameters of the mapping schema created by those algorithms:
\begin{enumerate}
  \item\emph{Number of reducers}. This is the number of reducers used by the mapping schema to send all inputs to.
\item \emph{The communication cost}, $c$. The communication cost is defined to be the sum of all the bits that are required, according to the mapping schema, to transfer from the map phase to the reduce phase.
\end{enumerate}
Table~\ref{table:Bounds} summarizes all the results in this paper. Before describing the algorithms, we look at lower bounds for the above parameters as they are expressed in terms of the reducer capacity $q$ and sum of sizes of all inputs $s$.

\medskip
\begin{theorem}\label{th:a2a_communication_cost}
\textnormal{\textsc{(Lower bounds on the communication cost and number of reducers)}} For a list of inputs and a given reducer capacity $q$, the communication cost and the number of reducers, for the \textit{A2A mapping schema problem}, are at least $\frac{s^2}{q}$ and $\frac{s^2}{q^2}$, respectively, where $s$ is the sum of all the input sizes.
\end{theorem}

\begin{proof}
Since an input $i$ is replicated to at least $\big\lfloor\frac{s-w_i}{q-w_i}\big\rfloor$ reducers, the communication cost for the input $i$ is $w_i\times \lfloor\frac{s-w_i}{q-w_i}\big\rfloor$. Hence, the communication cost for all the inputs will be at least $\sum_{i=1}^ m w_i \frac{s-w_i}{q-w_i}$. Since
$s\geq q$, we can conclude $\frac{s-w_i}{q-w_i}\geq\frac{s}{q}$. Thus, the communication cost is at least $\sum_{i=1}^ m w_i \frac{s}{q} = \frac{s^2}{q}$.

Since the communication cost, the number of bits to be assigned to reducers, is at least $\frac{s^2}{q}$, and a reducer can hold inputs whose sum of the sizes is at most $q$, the number of reducers must be at least $\frac{s^2}{q^2}$.
\end{proof}

\begin{table}[t]

\begin{center}
\bgroup
\def\arraystretch{1.75}
\scriptsize
   \begin{tabular}{ | p{5.5cm} |  l | l | l | } \hline

    Cases & Theorems & Communication cost & Approximation ratio\\ \hline \hline

    \multicolumn{4}{c}{The \textbf{lower bounds} for the \textit{A2A mapping schema problem}} \\\hline

    Different-sized inputs &~\ref{th:a2a_communication_cost} & $\frac{s^2}{q}$ & \\ \hline

    Equal-sized inputs &~\ref{th:a2a_2-step-The total communication cost equal input} & $m\big\lfloor\frac{m-1}{q-1}\big\rfloor$ & \\ \hline

    \multicolumn{4}{c}{The \textbf{lower bounds} for the \textit{X2Y mapping schema problem}} \\\hline

    Different-sized inputs &~\ref{th:x2y_communication_cost}  & $\frac{2\cdot sum_x \cdot sum_y}{q}$ &\\ \hline

    \multicolumn{4}{c}{\textbf{Optimal} algorithms for the \emph{A2A mapping schema problem} ($^{\ast}$ equal-sized inputs) } \\\hline

    Algorithm for reducer capacity $q = 2$  & \ref{opt-thm4} &$m(m-1)$ & optimal \\ \hline

    Algorithm for reducer capacity $q = 3$ & \ref{opt-thm4} &$m\big\lfloor\frac{m-1}{2}\big\rfloor$ & optimal \\ \hline

    The \emph{AU} method: When $q$ is a prime number &  \ref{opt-thm4} & $m\big\lfloor\frac{m-1}{q-1}\big\rfloor$& optimal \\ \hline

    \multicolumn{4}{c}{Non-optimal algorithms for the \emph{A2A mapping schema problem} and their \textbf{upper bounds}} \\ \hline

    Bin-packing-based algorithm, not including an input of size $> \frac{q}{2}$ &~\ref{th:our_bounds} & $\frac{4s^2}{q}$ & $\frac{1}{4}$ \\ \hline

    Algorithm~\ref{alg:2stepmethodforoddq} & ~\ref{th:a2a_2-step-The total communication cost1} & $\frac{q}{2k}\big\lceil\frac{sk}{q(k-1)}\big\rceil(\big\lceil\frac{sk}{q(k-1)}\big\rceil-1)$ & $1/k-1$\\ \hline

    Algorithm 2: The first extension of the \emph{AU method} &~\ref{th:first ext to the au method reducer_communication_cost} & $qp(p+1)+z^{\prime}$ & $q/(q+1)$\\ \hline

    Algorithm 3: The second extension of the \emph{AU method} & ~\ref{th:second ext to the au method reducer_communication_cost} & $q^2\times(q(q+1))^{l-1}$ & $(q^l-1)/q(q-1)(q+1)^{l-1}$ \\ \hline

Bin-packing-based algorithm considering an input of size $> \frac{q}{2}$ & ~\ref{th:a2a_big_input}  & $(m-1)\cdot q+\frac{4s^2}{q}$ & $\frac{s^2}{mq^2}$ \\ \hline

    \multicolumn{4}{c}{A non-optimal algorithm for the \emph{X2Y mapping schema problem} and their \textbf{upper bounds}} \\ \hline

    Bin-packing-based algorithm, $q=2b$ & ~\ref{th:our_bounds_x2y}  & $\frac{4\cdot sum_x \cdot sum_y}{b}$ & $\frac{1}{4}$\\ \hline

 \multicolumn{4}{|p{14cm}|}{\textit{Approximation ratio.} The ratio between the optimal communication cost and the communication cost obtained from an algorithm.

 Notations: $s$: sum of all the input sizes. $q$: the reducer capacity. $m$: the number of inputs. $\mathit{sum}_x$: sum of input sizes of the list $X$. $\mathit{sum}_y$: sum of input sizes of the list $Y$. $p$: the nearest prime number to $q$. $l>2$. $k>1$.
 } \\\hline

    \end{tabular}
    \egroup
    \caption{The bounds for heuristics for the \emph{A2A} and the \emph{X2Y} mapping schema problems.}
\label{table:Bounds}

    \end{center}
\end{table}

\begin{table}[t]
\begin{center}
\bgroup
\def\arraystretch{1.75}
\scriptsize
   \begin{tabular}{ | l |  l | } \hline

    Algorithms &  Inputs \\ \hline

\multicolumn{2}{c}{Non-optimal algorithms for the \emph{A2A mapping schema problem}} \\ \hline

Bin-packing-based algorithm &  Any number of inputs of any size \\\hline

Algorithm~\ref{alg:2stepmethodforoddq} &  Any number of inputs of size at most $\frac{q}{k}$, $k>3$\\ \hline

Algorithm 2: The first extension of the \emph{AU method} & $p^2 + p \cdot l + l$, $p+l=q$,  $l>2$\\ \hline

Algorithm 3: The second extension of the \emph{AU method} & $q^l$, $l > 2$ and $q$ is a prime number\\ \hline

\multicolumn{2}{c}{A non-optimal algorithm for the \emph{X2Y mapping schema problem}} \\ \hline

Bin-packing-based algorithm,  $>\frac{q}{2}$  & Any number of inputs of any size \\\hline

\multicolumn{2}{|p{14cm}|}{Notations: $w_i$ and $w_j$: the two largest size inputs of a list. $p$: the nearest prime number to $q$. $w_k$: the largest input of a list $X$. $w_k^{\prime}$: the largest input of a list $Y$.} \\\hline

    \end{tabular}
    \egroup
    \caption{Reducer capacity and input constraints for different algorithms for the mapping schema problems.}
    \label{table:Reducer capacity and input constraints for different algorithms for the mapping schema problems}
\end{center}
\end{table}

\medskip
\subsection{Bin-packing-based Approximation}
\label{subsubsec:Different sized inputs}
Our general strategy for building approximation algorithms is as follows: we use a known bin-packing algorithm to place the given $m$ inputs to bins of size $\frac{q}{k}$, $k\geq 2$. Assume that we need $x$ bins to place $m$ inputs. Now, each of these bins is considered as a single input of size $\frac{q}{k}$ for our problem of finding an optimal mapping schema. Of course, the assumption is that all inputs are of size at most $\frac{q}{k}$, $k\geq 2$.

First-Fit Decreasing (FFD) and Best-Fit Decreasing (BFD)~\cite{Coffman:1996:AAB:241938.241940} are most notable bin-packing algorithms. FFD or BFD bin-packing algorithm ensures that all the bins (except only one bin) are at least half-full. There also exists a pseudo polynomial bin-packing algorithm, suggested by Karger and Scott~\cite{DBLP:conf/approx/KargerS08}, that can place the $m$ inputs in as few bins as possible of certain size.

\begin{example}
Let us discuss in more detail the case $k=2$. In this case, since the reducer capacity is $q$, any two bins can be assigned to a single reducer. Hence, the approximation algorithm uses at most $\frac{x(x-1)}{2}$ reducers, where $x$ is the number of bin; see Figure~\ref{fig:Bin-packing-based heuristic} for an example.
\end{example}

\begin{figure}[h]
\begin{center}
 \includegraphics[scale=1.2]{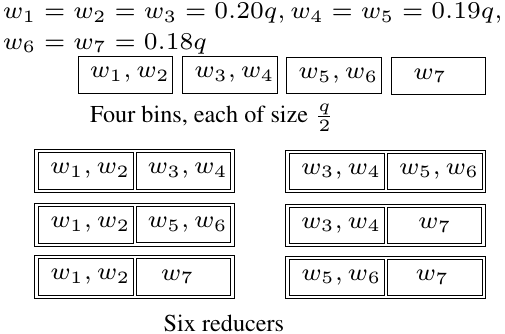}
 \end{center}\B
 \caption{Bin-packing-based approximation algorithm.}
 \label{fig:Bin-packing-based heuristic}
\end{figure}

For this strategy a lower bound on communication cost depends also on $k$ as follows:

\medskip
\begin{theorem}[Lower bound on the communication cost]
\label{th:a2a_2-step-The total communication cost}
Let $q>1$ be the reducer capacity, and let $\frac{q}{k}$, $k>1$, is the bin size. Let the sum of the given inputs is $s$. The communication cost, for the \textit{A2A mapping schema problem}, is at least $s\big\lfloor\frac{\frac{sk}{q}-1}{k-1}\big\rfloor $.
\end{theorem}

\begin{proof}
A bin can hold inputs whose sum of the sizes is at most $\frac{q}{k}$. Since the total sum of the sizes is $s$, it is required to divide the inputs into at least $x=\frac{sk}{q}$ bins. Now, each bin can be considered as an identical sized input.

Since a bin $i$ is required to be sent to at least $\big\lfloor\frac{x-1}{k-1}\big\rfloor$ reducers (to be paired with all the other bins), the sum of the number of copies of ($x$) bins sent to reducers is at least $x\big\lfloor\frac{x-1}{k-1}\big\rfloor$. We need to multiply this by $\frac{q}{k}$ (the size of each bin) to find the communication cost. Thus, we have at least
$$x\Big\lfloor\frac{x-1}{k-1}\Big\rfloor \frac{q}{k}= \frac{sk}{q}\Big\lfloor\frac{\frac{sk}{q}-1}{k-1}\Big\rfloor \frac{q}{k} = s\Big\lfloor\frac{\frac{sk}{q}-1}{k-1}\Big\rfloor$$
communication cost.
\end{proof}
This communication cost in the above theorem, as expected, is larger than the one in Theorem~\ref{th:a2a_communication_cost}, where no restriction in a specific strategy was taken into account.

\begin{example}
\textit{Example for $k=2$.}
Let us apply our strategy to the case where $k=2$, \textit{i}.\textit{e}., we have the algorithm: (\textit{i}) we do bin-packing to put the inputs in bins of size $\frac{q}{2}$; and (\textit{ii}) we provide a mapping schema for assigning each pair of bins to at least one reducer. Such a schema is easy and has been discussed in the literature (\textit{e}.\textit{g}., \cite{DBLP:journals/crossroads/Ullman12}).
\end{example}

FFD and BFD bin-packing algorithms provide an $\frac{11}{9}\cdot\textsc{Opt}$ approximation ratio~\cite{johnson1973near}, \textit{i}.\textit{e}., if any optimal bin-packing algorithm needs $\textsc{Opt}$ bins to place ($m$) inputs in the bins of a given size $\frac{q}{2}$, then FFD and BFD bin-packing algorithms always use at most $\frac{11}{9}\cdot\textsc{Opt}$ bins of an identical size (to place the given $m$ inputs). Since we require at most $\frac{x(x-1)}{2}$ reducers for a solution to the \emph{A2A mapping schema problem}, the algorithm requires at most ${(\frac{11}{9}\cdot\textsc{Opt})}^2/2$ reducers.

Note that, here in this case, \textsc{Opt} does not indicate the optimal number of reducers to assign $m$ inputs that satisfy the \emph{A2A mapping schema problem}; \textsc{Opt} indicates the optimal number of bins of size $\frac{q}{2}$ that are required to place $m$ inputs.

The following theorem gives the upper bounds that this approximation algorithm achieves on the communication cost and the number of reducers.

\medskip\begin{theorem}
\label{th:our_bounds}
\textnormal{\textsc{(Upper bounds on communication cost and number of reducers for $k=2$)}} The above algorithm using a bin size $b=\frac{q}{2}$ where $q$ is the reducer capacity achieves the following upper bounds: the number of reducers and the communication cost, for the \textit{A2A mapping schema problem}, are at most $\frac{8s^2}{q^2}$, and at most $4\frac{s^2}{q}$, respectively, where $s$ is the sum of all the input sizes.
\end{theorem}
\begin{proof}
A bin $i$ can hold inputs whose sum of the sizes is at most $b$. Since the total sum of the sizes is $s$, it is required to divide the inputs into at least $\frac{s}{b}$ bins. Since the FFD or BFD bin-packing algorithm ensures that all the bins (except only one bin) are at least half-full, each bin of size $\frac{q}{2}$ has at least inputs whose sum of the sizes is at least $\frac{q}{4}$. Thus, all the inputs can be placed in at most $\frac{s}{q/4}$ bins of size $\frac{q}{2}$. Since each bin is considered as a single input, we can assign every two bins to a reducer, and hence, we require at most $\frac{8s^2}{q^2}$ reducers.
Since each bin is replicated to at most $4\frac{s}{q}$ reducers, the communication cost is at most $\sum_{1\leq i\leq m} w_i \times 4\frac{s}{q} = 4\frac{s^2}{q}$.
\end{proof}

\medskip
\section{Equal-Sized Inputs Optimal Algorithms}
\label{sec:Equal-Sized Inputs Optimal Algorithms}
As we explained, looking at inputs of same size makes sense because we imagine the inputs are being bin-packed into bins of size $\frac{q}{k}$, for $k\geq 2$ (using bin-packing-based algorithm Section~\ref{subsubsec:Different sized inputs}), and that once this is done, we can treat the bins themselves as things of unit size to be sent to the reducers. Thus, in this section, we will shift the notation so that all inputs are of unit size, and $q$ is some small integer, \textit{e}.\textit{g}., 3.

In this section, we provide optimal algorithms for $q=2$ (in Section~\ref{q2-sec}) and $q=3$ (in Section~\ref{q3-sec}). Afrati and Ullman~\cite{DBLP:conf/ideas/AfratiU13} provided an optimal algorithm for the \emph{A2A mapping schema problem} where $q$ is a prime number and the number of inputs is $m=q^2$. We extend this algorithm for $m=q^2+q+1$ inputs (in Section~\ref{subsec:AU method}), and this extension also meets the lower bound on the communication cost. We will generalize these three algorithms in the Sections~\ref{subsec:Generalizing the Technique from section q is equal to 3} and~\ref{subsec:Generalizing Techniques from the AU method}.

In this setting, by minimizing the number of reducers, we minimize communication, since each reducer is more-or-less filled to capacity. So we define
\begin{itemize}
\item $r(m,q)$ to be the minimum number of reducers of capacity q that can solve the all-pairs problem for m inputs.
\end{itemize}

The following theorem sets a lower bound on $ r(m,q)$ and the communication cost for this setting.

\medskip
\begin{theorem}\label{th:a2a_2-step-The total communication cost equal input}
\textnormal{\textsc{(Lower bounds on the communication cost and number of reducers)}} For a given reducer capacity $q>1$ and a list of $m$ inputs, each input is of size one, the communication cost and the number of reducers ($r(m,q)$), for the \textit{A2A mapping schema problem}, are at least $m\big\lfloor\frac{m-1}{q-1}\big\rfloor$ and at least $\big\lfloor\frac{m}{q}\big\rfloor\big\lfloor\frac{m-1}{q-1}\big\rfloor$, respectively.
\end{theorem}
\begin{proof}
Since an input $i$ is required to be sent to at least $\big\lfloor\frac{m-1}{q-1}\big\rfloor$ reducers, the sum of the number of copies of ($m$) inputs sent to reducers is at least $m\big\lfloor\frac{m-1}{q-1}\big\rfloor$, which result in at least $m\big\lfloor\frac{m-1}{q-1}\big\rfloor$ communication cost.

There are at least $m\big\lfloor\frac{m-1}{q-1}\big\rfloor$ number of copies of ($m$) inputs to be sent to reducers and a reducer can hold at most $q$ inputs; hence, $r(m,q) \geq \big\lfloor\frac{m}{q}\big\rfloor\big\lfloor\frac{m-1}{q-1}\big\rfloor$.
\end{proof}

\medskip
\subsection{Reducer Capacity $q=2$}
\label{q2-sec}
Here, we offer a recursive algorithm and show that this algorithm does not only obtain the bound $r(m,2) \leq \frac{m(m-1)}{2}$, but it does so in a way that divides the reducers into $m-1 $ ``teams'' of $\frac{m}{2}$ reducers, where each team has exactly one occurrence of each input. We will use these properties of the output of this algorithm to build an algorithm for $q=3$ in the next subsection.

\medskip\medskip\noindent\textbf{The recursive algorithm.} We are given a list $A$ of $m$ inputs. The intention is to have all pairs of inputs from list $A$ partitioned into $m-1$ teams with each team containing exactly $\frac{m}{2}$ pairs and each input appearing exactly once within a team. Hence, we will use $\frac{m(m-1)}{2}$ reducers for assigning pairs of each input.

We split $A$ into two sublists $A_1$ and $A_2$ of size $\frac{m}{2}$ each. Suppose, we have the $\frac{m}{2}-1$ teams for a list of size $\frac{m}{2}$. We will take the $\frac{m}{2}-1$ teams of $A_1$, the $\frac{m}{2}-1$ teams of $A_2$ and ``mix them up'' in a rather elaborate way to form the $m-1$ teams for $A$:

Let the teams for $A_1$ and $A_2$ be $\{g_1,g_2,g_3, \ldots, g_{\frac{m}{2}}\}$ and $\{h_1,h_2,h_3, \ldots, h_{\frac{m}{2}}\}$ respectively. We will form two kind of teams, teams of kind I and teams of kind II as follows:

{\sl Teams of kind I.} We will form $\frac{m}{2}$ teams of kind I by taking one input from $A_1$ and one input from $A_2$. For example, the first team for $A$ is $\{(g_1,h_1),(g_2,h_2), (g_3,h_3), \ldots, (g_{\frac{m}{2}},h_{\frac{m}{2}})\}$, the second team for $A$ is $\{(g_1,h_2),(g_2,h_3), (g_3,h_4), \ldots, (g_{\frac{m}{2}},h_1)\}$, and so on.

{\sl Teams of kind II.} We will form the remaining $\frac{m}{2}-1$ teams having $\frac{m}{2}$ reducers in each. In teams of kind I each pair (reducer) contains only inputs from one of the lists $A_1$ or $A_2$. Now we produce pairs, with each pair having both inputs from $A_1$ or $A_2$. In order to do that, we divide recursively divide $A_1$ into two sublists and perform the operation what we performed in the team of kind I. The same procedure is recursively implemented on $A_2$.

\begin{example}
For $m=8$, we form 7 teams. First we form teams of kind I. We divide 8 inputs into two lists $A_1$ and $A_2$. After that, we take one input from $A_1$ and one input from $A_2$, and create 4 teams, see Figure~\ref{fig:team for 8 inputs}. Now, we recursively follow the same rule on each sublist, $A_1$ and $A_2$, and create 3 remaining teams of kind II, see Figure~\ref{fig:team for 8 inputs}.

\begin{figure}[!h]
\centering
\begin{minipage}[b]{0.45\linewidth}
\centering
1,5  ~~~~~~~~  1,6  ~~~~~~~  1,7   ~~~~~~~ 1,8\\
2,6   ~~~~~~~ 2,7   ~~~~~~~ 2,8   ~~~~~~~ 2,5\\
3,7   ~~~~~~~ 3,8   ~~~~~~~ 3,5  ~~~~~~~ 3,6\\
4,8    ~~~~~~~ 4,5   ~~~~~~~ 4,6 ~~~~~~~   4,7 \\
Team 1 ~~ Team 2 ~~ Team 3 ~~ Team 4
\subcaption{Teams of kind I}
\end{minipage}
\begin{minipage}[b]{0.48\linewidth}
\centering
1,3 ~~~~~~~~ 1,4 ~~~~~~~~ 1,2\\
2,4 ~~~~~~~~ 2,3 ~~~~~~~~ 3,4\\
5,7 ~~~~~~~~ 5,8 ~~~~~~~~ 5,6\\
6,8 ~~~~~~~~ 6,7 ~~~~~~~~ 7,8\\
Team 5 ~~ Team 6 ~~ Team 7

     \subcaption{Teams of kind II}
    \end{minipage}

\caption{The teams for $m=8$ and $q=2$.}
\label{fig:team for 8 inputs}
\end{figure}

Actually in Figure~\ref{fig:The 2-step method for q 3}, the teams for this example are shown in non-bold face fonts (two in each triplet in Figure~\ref{fig:The 2-step method for q 3}, notice that they are from 1-8) in teams 1 through 7 in Figure~\ref{fig:The 2-step method for q 3}.

\end{example}
The following theorem is easy to prove.

\medskip
\begin{theorem}
In each team an input appears only once. In each team all inputs appear. There are $m-1$ teams which is the minimum possible. Hence this is an optimal mapping scheme that assigns inputs to reducers.
\end{theorem}

This works if the number of inputs is a power of two. We can use known techniques to make it work with good approximation in general.

\medskip
\subsection{Reducer Capacity $q=3$}
\label{q3-sec}
Here, we present an algorithm that constructs an optimal mapping schema for $q=3$. Our recursive algorithm starts by taking the mapping schema constructed in previous subsection for $q=2$. We showed there that for $q=2$, we can not only obtain the bound $r(m,2) \leq \frac{m(m-1)}{2}$, but that we can do so in a way that divides the reducers into $m-1$ teams of $\frac{m}{2}$ reducers in each team, where each team has exactly one occurrence of each input.

Now, we split $m$ inputs into two disjoint sets: set $A$ and set $B$. Suppose $m=2n-1$. Set $A$ has $n$ inputs and set $B$ has $n-1$ inputs. We start with the $n$ inputs in set $A$, and create $n-1$ teams of $\frac{n}{2}$ reducers, each reducer getting two of the $n$ inputs in $A$, by following the algorithm given in Section~\ref{q2-sec}. Next, we add to all reducers in one team another input from set $B$. \textit{I}.\textit{e}., in a certain team we add to all $\frac{n}{2}$ reducers of this team a certain input from set $B$, and thus, we form a triplet for each reducer.

Since there are $n-1$ teams, we can handle another $n-1$ inputs. This is the start of a solution for $q=3$ and $m=2n-1$ inputs. To complete the solution, we add the reducers for solving the problem for the $n-1$ inputs of the set $B$. That leads to the following recurrence

\begin{align*}
r(m, 3) &= \frac{n(n-1)}{2} + r(n-1, 3),\: \textnormal{where}\: m=2n-1\\
r(3,3) &= 1
\end{align*}

We solve the recurrence for $m$ a power of 2, and it exactly matches the lower bound of $r(m,3) = \frac{m(m-1)}{6}$. Moreover, notice that we can prove that this case is optimal either by proving that $r(m,3) = m(m-1)/6$ (as we did above) or by observing that every pair of inputs meets exactly in one reducer. This is easy to prove. Hence the following theorem:

\medskip\begin{theorem}
This algorithm constructs an optimal mapping schema for the reducer capacity $3$.
\end{theorem}

\begin{figure*}[h]
\begin{center}
  \begin{minipage}{.99\textwidth}
  \centering
  \includegraphics{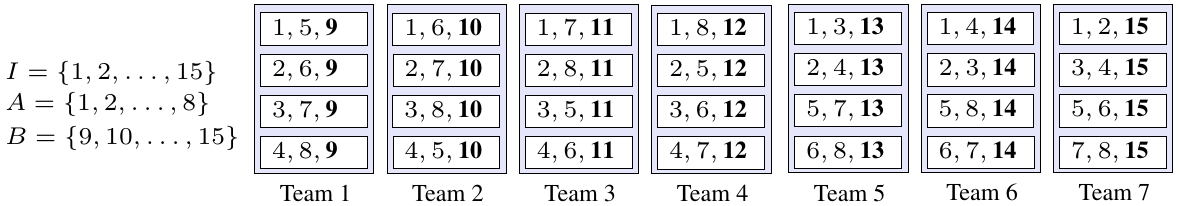}
  \label{fig:2-step-method-odd-q-aa}
  \end{minipage}
  \begin{minipage}{.99\textwidth}
  \centering
  \includegraphics{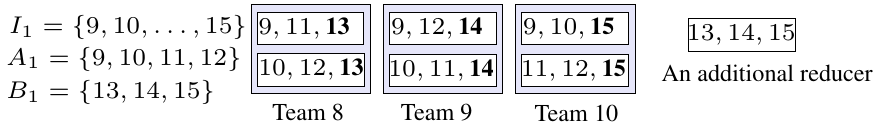}
  \label{fig:2-step-method-odd-q-bb}
  \end{minipage}
\end{center}
\B
\caption{An example of a mapping schema for $q=3$ and $m=15$.}
\label{fig:The 2-step method for q 3}
\end{figure*}

\begin{example}
An example is shown in Figure~\ref{fig:The 2-step method for q 3}. We explained how this figure is constructed for $q=2$ (the non-bold entries). Now we use the algorithm just presented here to construct the 35 ($=15 \times \frac{14}{6}$) reducers. We explain below in detail how we construct these 35 reducers.

We are given 15 inputs ($I=\{1,2,\ldots,15\}$). We create two sets, namely $A$ of $y=8$ inputs and $B$ of $x=7$ inputs, and arrange $(y-1)\times\big\lceil \frac{y}{2} \big\rceil= 28$ reducers in the form of $7$ \emph{teams} of $4$ reducers in each team. These 7 teams assign each input of the set $A$ with all other inputs of the set $A$ and all the inputs in the set $B$ as follows. We pair every two inputs of the set $A$ and assign them to exactly one of 28 reducers as we explained in Section~\ref{q2-sec}. Once every pair of $y=8$ inputs of the set $A$ is assigned to exactly one of 28 reducers, then we assign the $i^{th}$ input of the set $B$ to all the four reducers of $(i-8)^{th}$ team. Thus, \textit{e}.\textit{g}., input 10 is assigned to the four
reducers of Team 2.

Now these 28 reducers have seen that each pair of inputs from set $A$ meet in at least one reducer and each pair of inputs, one from $A$ and one from $B$ meet in at least one reducer. Thus, it remains to build more reducers so that each pair of inputs (both) from set $B$ meet. According to the recursion we explained, we break set $B$ into sets $A_1$ and $B_1$, of size 4 and 3 respectively, and we apply our method again. In particular, we create two sets, $A_1=\{9,10,11,12\}$ of $y_1=4$ inputs and $B_1=\{13,14,15\}$ of $x_1=3$. Then, we arrange $(y_1-1)\times\big\lceil \frac{y_1}{2} \big\rceil= 6$ reducers in the form of $3$ \emph{teams} of $2$ reducers in each team. We assign each pair of inputs of the set $A_1$ to these 6 reducers, and then $i^{th}$ input of the set $B_1$ to all the two reducers of a team, see Team 8 to Team 10.

The last team is constructed so that all inputs in $B_1$ meet at the same reducers (since $B_1$ has only 3 elements and 3 is the size of a reducer, one reducer suffices for this to happen).
\end{example}

\medskip\noindent{\bf Open problem.} Now the interesting observation is that if we can argue that the resulting reducers can be divided into $\frac{m-1}{2}$ teams of $\frac{m}{3}$ reducers each (with each team having one occurrence of each input), then we can extend the idea to $q=4$, and perhaps higher.

\medskip
\subsection{When $q$ or $q-1$ is a prime number}
\label{subsec:AU method}
An algorithm to provide a mapping schema for the reducer capacity $q$, where $q$ is a prime number, and $m=q^2$ inputs is suggested by Afrati and Ullman in~\cite{DBLP:conf/ideas/AfratiU13}. This method meets the lower bounds on the communication cost. We call this algorithm the \emph{AU method}. For the sake of completeness, we provide an overview of the \emph{AU method}. Interested readers may refer to~\cite{DBLP:conf/ideas/AfratiU13}.

\medskip\noindent\textbf{The \emph{AU method}.} We divide the $m$ inputs into $q^2$ equal-sized \emph{subsets} (each with $\frac{m}{q^2}$ inputs) that are arranged in a $Q=q\times q$ square. The subsets in row $i$ and column $j$ are represented by $S_{i,j}$, where $0\leq i<q$ and $0\leq j<q$.

We now organize $q(q+1)$ reducers in the form of $q+1$ \emph{teams} of $q$ \emph{players} (or reducers) in each team. Note that sum of sizes of the inputs in each row and column of the $Q$ square is exactly $q$.

The teams are arranged from 0 to $q$, and the reducers are arranged from 0 to $q-1$. We first arrange inputs to the team $q$. Since the sum of the sizes in each column of the $P$ square is $q$, we place one column of the $P$ square to one reducer of the team $q$. Now we place the inputs to the remaining teams. We use modulo operation for the assignment of each subset to each team. The subset $S_{i, j}$ is assigned to a reducer $r$ of each team $t$, $0\le t < q$, such that $(i+tj)modulo \: q=r$. An example for $q=3$ and $m=9$ is given in Figure~\ref{fig:AU-q-3-m-9}.

\begin{figure*}[h]
  \centering
  \includegraphics[width=150mm]{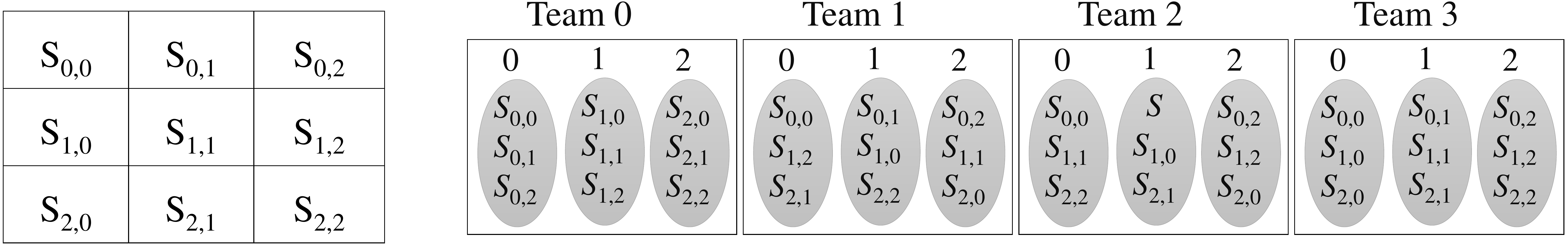}
  \caption{The \emph{AU method} for the reducer capacity $p=3$ and $m=9$.}
\label{fig:AU-q-3-m-9}
\end{figure*}

\medskip\medskip \noindent \emph{Total required reducers.} The \emph{AU method} uses $q(q+1)$ reducers, which are organized in the form of $q+1$ teams of $q$ reducers in each team, and the communication cost is $q^2(q+1)$.

\medskip\medskip\noindent\textbf{A simple extension of the \emph{AU method}}. Now, we can extend the \emph{AU method} as follows: we can add $q+1$ additional inputs, add one to each reducer and add one more reducer that has the $q+1$ new inputs. That gives us reducers of size $q=q+1$ and $m = q^2 + q + 1$, or $r(q^2 + q + 1, q+1) = q(q+1) + 1 = q^2 + q + 1$. If you substitute $m = q^2 + q + 1$ and $p=p+1$, you can check that this also meets the bound of $r = \frac{m(m-1)}{q(q-1)}$. In Figure~\ref{fig:au_example for q plus one}, we show a mapping schema for this extension to the \emph{AU method} for $q=4$ and $m=14$.

\begin{figure*}[h]
\begin{center}
  \begin{minipage}{.99\textwidth}
  \centering
  \includegraphics{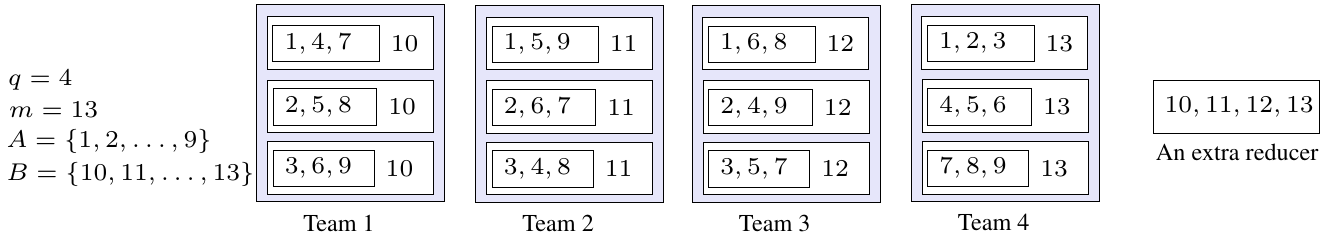}
  \end{minipage}
\end{center}
\caption{An optimum mapping schema for $q=4$ and $m=14$ by extending the \emph{AU method}.}
\label{fig:au_example for q plus one}
\end{figure*}

\medskip In conclusion, in this section we have shown the following:
\medskip\begin{theorem}
\label{opt-thm4}
We can construct optimal mapping schemas for the following cases:
\begin{enumerate}
\item $q=2$.
\item $q=3$.
\item $q$ being a prime number and $m=q^2$.
\item $q-1$ being a prime number and $m=(q-1)^2+q$, where $q$ is the reducer capacity and $m$ is the number of inputs.
\end{enumerate}
\end{theorem}

\medskip \noindent{\bf Open problem:} Can we generalize the last idea to get optimal schemas for more cases?

\medskip\medskip\noindent\textbf{Approximation Algorithms for the \emph{A2A Mapping Schemas Problem}.} We can use the optimal mapping schemas of Section~\ref{sec:Equal-Sized Inputs Optimal Algorithms} to construct good approximation of mappings schemas in many cases. The general techniques, we will use in this section move along the following dimensions/ideas:

\begin{itemize}
\item Assuming that there are no inputs of size greater than $\frac{q}{k}$, construct bins of size $\frac{q}{k}$, and treat each of the bins as a single input of size 1 and assume the reducer capacity is $k$. Then apply one of the optimal techniques of Section~\ref{sec:Equal-Sized Inputs Optimal Algorithms} to construct a mapping schema. These algorithms are presented in Sections~\ref{subsec:Generalizing the Technique from section q is equal to 3} and~\ref{subsec:A hybrid approach}.

\item Getting inspiration from the methods developed (or only presented -- in the case of the \emph{AU method}) in Section~\ref{subsec:AU method}, we extend the ideas to construct good approximation algorithms for inputs that are all of equal size (see Sections~\ref{subsubsec:When m is anything and we take the nearest prime to m} and~\ref{subsec:second ext to the au method}).
\end{itemize}

Thus, in Sections~\ref{subsec:Generalizing the Technique from section q is equal to 3},~\ref{subsec:Generalizing Techniques from the AU method}, and~\ref{subsec:A hybrid approach}, we will give several such techniques and show that some of them construct mapping schemas close to the optimal. To that end, we have already shown a schema based on bin-packing algorithms in Section~\ref{subsubsec:Different sized inputs}.

\medskip
\section{Generalizing the Technique for the Reducer Capacity ${q}>3$ and Inputs of Size $\leq$ $q/k$, $k>3$}
\label{subsec:Generalizing the Technique from section q is equal to 3}

In this section, we will generalize the algorithm for $q=3$ given in Section~\ref{q3-sec} and present an algorithm (Algorithm~\ref{alg:2stepmethodforoddq}) for inputs of size less than or equal to $\frac{q}{k}$ and $k>3$. For simplicity, we assume that $k$ divides $q$ evenly throughout this section.

\medskip
\subsection{Algorithm~\ref{alg:2stepmethodforoddq}A}

We divide Algorithm~\ref{alg:2stepmethodforoddq} into two parts based on the value of $k$ as even or odd. Algorithm~\ref{alg:2stepmethodforoddq}A considers that $k$ is an odd number. Pseudocode of Algorithm~\ref{alg:2stepmethodforoddq}A is given in Appendix~\ref{app:algo_correctness_1}. Algorithm~\ref{alg:2stepmethodforoddq}A works as follows:

First places all the given inputs, say $m^{\prime}$, to some bins, say $m$, each of size $\frac{q}{k}$, $k>3$ is an odd number. Thus, a reducer can hold an odd number of bins. After placing all the $m^{\prime}$ inputs to $m$ bins, we can treat each of the $m$ bins as a single input of size one and the reducer capacity to be $k$. Now, it is easy to turn the problem to a case similar to the case of $q=3$. Hence, we divide the $m$ bins into two sets $A$ and $B$, and follow a similar approach as given in Section~\ref{q3-sec}.

\medskip \noindent\textbf{Aside.} Equivalently, we can consider $q$ to be odd and the inputs to be of unit size. In what follows, we will continue to use $q$, which is an odd number, as the reducer capacity and assume all inputs (that are actually bins containing inputs) are of unit size.

\begin{example}
If $q=30$ and $k=5$, then we can pack given inputs to some bins of size $6$. Hence, a reducer can hold 5 bins. Equivalently, we may consider each of the bins as a single input of size 1 and $q=5$.
\end{example}

For understanding of Algorithm~\ref{alg:2stepmethodforoddq}A, an example for $q=5$ is presented in Figure~\ref{fig:The 2-step method for q 5}, where we obtain $m=23$ bins (that are considered as 23 unit-sized inputs) after implementing a bin-packing algorithm to given inputs.

\begin{figure*}[h]
\begin{center}
  \begin{minipage}{.99\textwidth}
  \centering
  \includegraphics{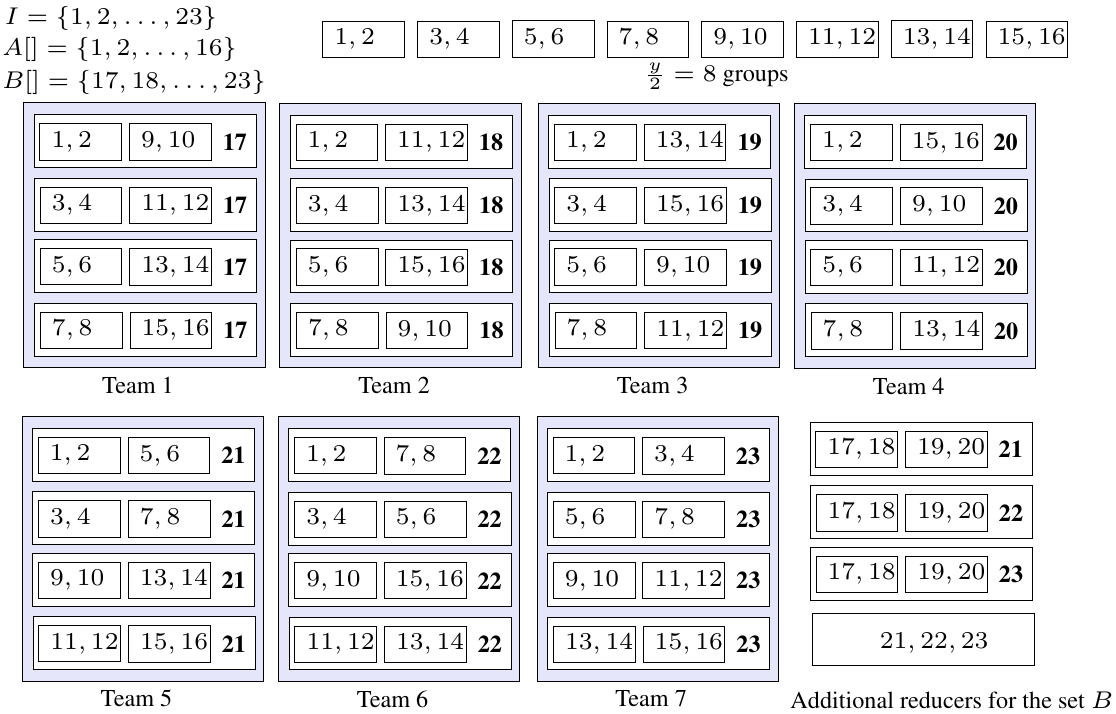}
  \label{fig:2-step-method-odd-q5-aa}
  \end{minipage}
\caption{Algorithm~\ref{alg:2stepmethodforoddq}A -- an example of a mapping schema for $q=5$ and 23 bins.}
\label{fig:The 2-step method for q 5}

\end{center}
\end{figure*}

Algorithm~\ref{alg:2stepmethodforoddq}A consists of six steps as follows:

\begin{enumerate}
\item \emph{Implement a bin-packing algorithm}: Implement a bin-packing algorithm to place all the given $m^{\prime}$ inputs to bins of size $\frac{q}{k}$, where $k>3$ is an odd number and the size of all the inputs is less than or equal to $\frac{q}{k}$. Let $m$ bins are obtained, and now each of the bins is considered as a single input.
  \item \emph{Division of bins (or inputs) to two sets, $A$ and $B$}: Divide $m$ inputs into two sets $A$ and $B$ of size $y = \big\lfloor \frac{q}{2} \big\rfloor (\big\lfloor \frac{2m}{q+1} \big\rfloor+1)$ and $x=m- y$, respectively.
  \item \emph{Grouping of inputs of the set $A$}: Group the $y$ inputs into $u = \big\lceil\frac{y}{q-\lceil q/2\rceil}\big\rceil$ disjoint groups, where each group holds $\big\lceil\frac{q-1}{2}\big\rceil$ inputs. (We consider each of the $u$ ($=\big\lceil\frac{y}{q-\lceil q/2\rceil}\big\rceil$) disjoint groups as a single input that we call the \textit{derived input}. By making $u$ disjoint groups\footnote{We suppose that $u$ is a power of 2. In case $u$ is not a power of 2 and $u>q$, we add dummy inputs each of size $\big\lceil\frac{q-1}{2}\big\rceil$ so that $u$ becomes a power of 2. Consider that we require $d$ dummy inputs. If groups of inputs of the set $B$ each of size $\big\lceil\frac{q-1}{2}\big\rceil$ are less than equal to $d$ dummy inputs, then we use inputs of the set $B$ in place of dummy inputs, and the set $B$ will be empty.} (or derived inputs) of $y$ inputs of the set $A$, we turn the case of any odd value of $q$ to a case where a reducer can hold only three inputs, the first two inputs are pairs of the derived inputs and the third input is from the set $B$.)
  \item \emph{Assigning groups (inputs of the set $A$) to some reducers}: Organize $(u-1)\times \big\lceil \frac{u}{2}\big\rceil$ reducers in the form of $u-1$ teams of $\big\lceil\frac{u}{2}\big\rceil$ reducers in each team. Assign every two groups to one of $(u-1)\times \big\lceil\frac{u}{2}\big\rceil$ reducers. To do so, we will prove the following Lemma~\ref{lm:recursive_def_odd}.
\begin{lemma}\label{lm:recursive_def_odd}
Let $q$ be the reducer capacity. Let the size of an input is $\big\lceil\frac{q-1}{2}\big\rceil$. Each pair of $u=2^i$, $i>0$, inputs can be assigned to $2^i - 1$ teams of $2^{i-1}$ reducers in each team.\footnote{The proof appears in Appendix~\ref{app:proof of th}.}
\end{lemma}
  \item \emph{Assigning inputs of the set $B$) to the reducers}: Once every pair of the derived inputs are assigned, then assign $i^{th}$ input of the set $B$ to all the reducers of $i^{th}$ team.
  \item \emph{Use previous steps on the inputs of the set $B$}: Apply (the above mentioned) steps 1-4 on the set $B$ until there is a solution to the \emph{A2A mapping schema problem} for the $x$ inputs.
\end{enumerate}

\medskip\begin{theorem}[The communication cost obtained using Algorithm~\ref{alg:2stepmethodforoddq}]\label{th:a2a_2-step-The total communication cost1}
For a given reducer capacity $q>1$, $k>3$, and a list of $m$ inputs whose sum of sizes is $s$, the communication cost, for the \textit{A2A mapping schema problem}, is at most $\frac{q}{2k}\big\lceil\frac{sk}{q(k-1)}\big\rceil(\big\lceil\frac{sk}{q(k-1)}\big\rceil-1)$.
\end{theorem}

\begin{proof}
Since the FFD or BFD bin-packing algorithm ensures that all the bins (except only one bin) are at least half-full, each bin of size $\frac{q}{k}$ has at least inputs whose sum of the sizes is at least $\frac{q}{k/2}$. Thus, all the inputs can be placed in at most $x=s/(q/(k/2))=\frac{sk}{2q}$ bins of size $\frac{q}{k}$. Now, each bin can be considered as an identical sized input.

According to the construction given in Algorithm~\ref{alg:2stepmethodforoddq}A, there are at most $g=\big\lceil\frac{2x}{k-1}\big\rceil$ groups (derived inputs) of the given $x$ bins. In order to assign each pair of the derived inputs, each derived input is required to assign to at most $g-1$ reducers. In addition, the size of each input (bin) is $\frac{q}{k}$, therefore we have at most
$$\frac{q}{k}\times g(g-1)/2=\frac{q}{k}\times \Big\lceil\frac{2x}{k-1}\Big\rceil\Big(\Big\lceil\frac{2x}{k-1}\Big\rceil-1\Big)/2$$
$$=\frac{q}{2k}\times \Big\lceil\frac{sk}{q(k-1)}\Big\rceil\Big(\Big\lceil\frac{sk/q}{k-1}\Big\rceil-1\Big)> \frac{4s^2}{q}$$
communication cost.
\end{proof}

\medskip\medskip \noindent \emph{Algorithm correctness.} The algorithm correctness appears in Appendix~\ref{app:algo_correctness_1}.

\medskip\medskip \noindent \textit{Approximation factor.} The optimal communication cost (from Theorem \ref{th:a2a_2-step-The total communication cost}) is $s\lfloor(\frac{sk}{q}-1)/k-1\rfloor \approx \frac{s^2}{q}\cdot\frac{k}{k-1}$ and the communication cost of the algorithm (from Theorem \ref{th:a2a_2-step-The total communication cost1}) is $\frac{q}{2k}\big\lceil\frac{sk}{q(k-1)}\big\rceil(\big\lceil\frac{sk}{q(k-1)}\big\rceil-1)\approx s^2k/q(k-1)^2$. Thus, the ratio between the optimal communication and the communication of our mapping schema is approximately $\frac{1}{k-1}$.

\medskip
\subsection{Algorithm~\ref{alg:2stepmethodforoddq}B}

For the sake of completeness, we include the pseudocode of the algorithm for handling the case when $k$ is an even number. We call it Algorithm~\ref{alg:2stepmethodforoddq}B and pseudocode is given in Appendix~\ref{app:algo_correctness_2}. In this algorithm, we are given $m^{\prime}$ inputs of size less than or equal to $\frac{q}{k}$ and $k\geq 4$ is an even number.

Similar to Algorithm~\ref{alg:2stepmethodforoddq}A, Algorithm~\ref{alg:2stepmethodforevenq}B first places all the $m^{\prime}$ inputs to $m$ bins, each of size $\frac{q}{k}$, $k>2$ is an even number. Thus, a reducer can hold an even number of bins. After placing all the $m^{\prime}$ inputs to $m$ bins, we can treat each of the $m$ bins as a single input of size one and the reducer capacity to be $k$. Now, we easily turn this problem to a case similar to the case of $q=2$. Hence, we divide the $m$ bins into two set $A$ and $B$, and follow a similar approach as given in Section~\ref{q2-sec}.

\begin{example}
If $q=30$ and $k=6$, then we can pack given inputs to some bins of size $5$. Hence, a reducer can hold 6 bins. Equivalently, we may consider each of the bins as a single input of size 1 and $q=6$.
\end{example}

\medskip\noindent\textbf{Note.} Algorithms~\ref{alg:2stepmethodforoddq}A and~\ref{alg:2stepmethodforoddq}B are based on a fact that how do we pack inputs in a \textit{well} manner to bins of even or odd size. To understand this point, consider $q=30$ and $m^{\prime}=46$. For simplicity, we assume that all the inputs are of size three. Now, consider $k=5$, so we will use 23 bins each of size $6$ and apply Algorithm~\ref{alg:2stepmethodforoddq}A. On the other, consider $k=6$, so we will use 46 bins each of size $5$ and apply Algorithm~\ref{alg:2stepmethodforoddq}B.

\medskip
\section{Generalizing the \emph{AU method}}
\label{subsec:Generalizing Techniques from the AU method}
In this section, we extend the \emph{AU method} (Section~\ref{subsec:AU method}) to handle more than $q^2$ inputs, when $q$ is a prime number, Algorithms 3 and 4. Recall that the \emph{AU method} can assign each pair of $q^2$ inputs to reducers of capacity $q$. We provide two extensions: (\textit{i}) take $m=p^2+p\cdot l+l$ identical-sized inputs and assign these inputs to reducers of capacity $p+l=q$, where $p$ is the nearest prime number to $q$, in Section~\ref{subsubsec:When m is anything and we take the nearest prime to m}, and (\textit{ii}) take $m=q^l$ inputs, where $l>2$, and assign inputs to reducers of capacity $q$, in Section~\ref{subsec:second ext to the au method}.

\medskip
\subsection{When we consider the nearest prime to $q$}
\label{subsubsec:When m is anything and we take the nearest prime to m}
We provide an extension to the \emph{AU method} that handles $m=p^2+p\cdot l+l$ identical-sized inputs and assigns them to reducers of capacity $p+l=q$, where $p$ is the nearest prime number to $q$. We call it the \emph{first extension to the AU method} (Algorithm 2).

\medskip\medskip\noindent \textbf{Algorithm 2: The \emph{First Extension of the AU method}.} We extend the \emph{AU method} by increasing the reducer capacity and the number of inputs. Consider that the \emph{AU method} assigns $p^2$ identical-sized inputs to reducers of capacity $p$, where $p$ is a prime number. We add $l(p+1)$ inputs and increase the reducer capacity to $p+l$ ($=q$).

In other words, $m$ identical-sized inputs and the reducer capacity $q$ are given. We select a prime number, say $p$, that is near most to $q$ such that $p+l=q$ and $p^2+l(p+1)\leq m$. Also, we divide the $m$ inputs into two disjoint sets $A$ and $B$, where $A$ holds at most $p^2$ inputs and $B$ holds at most $l(p+1)$ inputs.

Algorithm 2 consists of six steps, where $m$ inputs and the reducer capacity $q$ are inputs to Algorithm 2, as follows:
\begin{enumerate}
  \item Divide the given $m$ inputs into two disjoint sets $A$ of $y=p^2$ inputs and $B$ of $x=m-y$ inputs, where $p$ is the nearest prime number to $q$ such that $p+l=q$ and $p^2+l(p+1)\leq m$.
  \item Perform the \emph{AU method} on the inputs of the set $A$ by placing $y$ inputs to $p+1$ teams of $p$ bins in each team, where the size of each bin is $p$.
  \item Organize $p(p+1)$ reducers in the form of $p+1$ teams of $p$ reducers in each teams, and assign $j^{th}$ bin of $i^{th}$ team of bins to $j^{th}$ reducer of $i^{th}$ team of reducers.
  \item Group the $x$ inputs of the set $B$ into $u=\big\lceil \frac{x}{q-p}\big\rceil$ disjoint groups.
  \item Assign $i^{th}$ group to all the reducers of $i^{th}$ team.
  \item Use Algorithm~\ref{alg:2stepmethodforoddq}A or Algorithm~\ref{alg:2stepmethodforoddq}B to make each pair of inputs of the set $B$, depending on the case of the value of $q$, which is either an odd or an even number, respectively.
\end{enumerate}

Note that when we perform the above mentioned step 3, we assign each pair of inputs of the set $A$ to $p(p+1)$ reducers, and such an assignment uses $p$ capacity of each reducer. Now, each of $p(p+1)$ reducers has $q-p$ remaining capacity that is used to assign $i^{th}$ group of inputs of the set $B$. In this manner, all the inputs of the set $A$ are assigned with all the $m$ inputs.

\medskip\medskip \noindent \emph{Algorithm correctness.} The algorithm correctness appears in Appendix~\ref{app:algo_correctness_3}.

\medskip\begin{theorem}[The communication cost obtained using Algorithm 2]\label{th:first ext to the au method reducer_communication_cost}
Algorithm 2 requires at most $p(p+1)+z$ reducers, where $z=\frac{2l^2(p+1)^2}{q^2}$, and results in at most $qp(p+1)+z^{\prime}$ communication cost, where $z^{\prime}=\frac{2l^2(p+1)^2}{q}$, $q$ is the reducer capacity, and $p$ is the nearest prime number to $q$.
\end{theorem}
When $l=q-p$ equals to one, we have provided an extension of the \emph{AU method} in Section~\ref{subsec:AU method}, and in this case, we have an optimum mapping schema for $q$ and $m=q^2+q+1$ inputs.
\begin{proof}
In case of $l>1$, a single reducer cannot be used to assign all the inputs of the set $B$. Since Algorithm 2 is based on the \emph{AU method}, Algorithm~\ref{alg:2stepmethodforoddq}A, and Algorithm~\ref{alg:2stepmethodforoddq}B, we always use at most $p(p+1)+z$ reducers, where $z$ ($=\frac{2l^2(p+1)^2}{q^2}$) reducers are used to assign each pair of inputs of the set $B$ based on Algorithms~\ref{alg:2stepmethodforoddq}A or~\ref{alg:2stepmethodforoddq}B (for the value of $z$, the reader may refer to Theorem 11 of the technical report~\cite{A2A}). Thus, the communication cost is at most $qp(p+1)+z^{\prime}$, where $z^{\prime}$ ($=\frac{2l^2(p+1)^2}{q}$) is the maximum communication cost required by Algorithm~\ref{alg:2stepmethodforoddq}A or~\ref{alg:2stepmethodforoddq}B for assigning $(p+1)l$ inputs of the set $B$.
\end{proof}

\medskip\medskip \noindent \textit{Approximation factor.} The optimal communication cost using the \emph{AU method} is $q^2(q+1)$. Thus, the difference between the communication of our mapping schema ($q^2(q+1)+z^{\prime}$, when assuming $p$ is equal to $q$) and the optimal communication is $z^{\prime}$. We can see two cases, as follows:
\begin{enumerate}
  \item When $q$ is large. Consider that $q$ is greater than square or cube of the maximum difference between any two prime numbers. In this case, $z^{\prime}$ will be very small, and we will get almost optimal ratio.
  \item When $q$ is very small. In this case, then $z^{\prime}$ plays a role as follows: here, the number of inputs in the set $B$ will be at most $(p+1)l < q^2$. Thus, the ratio becomes $q/(q+1)$.
\end{enumerate}

\medskip
\subsection{For input size $m=q^l$ where $q$ is a prime number}
\label{subsec:second ext to the au method}
We also provide another extension to the \emph{AU method} that handles $m=q^l$ identical-sized inputs and assigns them to reducers of capacity $q$, where $q$ is a prime number and $l>2$. We call it the \emph{second extension to the AU method} (Algorithm 3).

\medskip\medskip \noindent \textbf{Algorithm 3: The \emph{Second Extension of the AU method}.} The second extension to the \emph{AU method} (Algorithm 3) handles a case when $m=q^l$, where $l>2$ and $q$ is a prime number. We present Algorithm 3 for $m=q^l$, $l>2$, inputs and the reducer capacity $q$, where $q$ is a prime number. Nevertheless, $m$ inputs that are less than but close to $q^l$ can also be handled by Algorithm 3 by adding dummy inputs such that $m=q^l$, $l>2$.

Algorithm 3 consists of two phases, as follows:

\medskip\noindent\textbf{\textit{The first phase: creation of a bottom up tree.}} Here, we present a simple example for the bottom-up tree's creation for $q=3$ and $m=3^4$; see Figure~\ref{fig:fig_second_ext_au_example}.

\begin{example}[Bottom-up tree creation]
A bottom-up tree for $m=q^l=3^4$ identical-sized inputs and $q=3$ is given in Figure~\ref{fig:fig_second_ext_au_example}. Here, we explain how we constructed it.

\begin{figure*}[h!]
\begin{center}
  \begin{minipage}{.99\textwidth}
  \centering
  \includegraphics[scale=0.20]{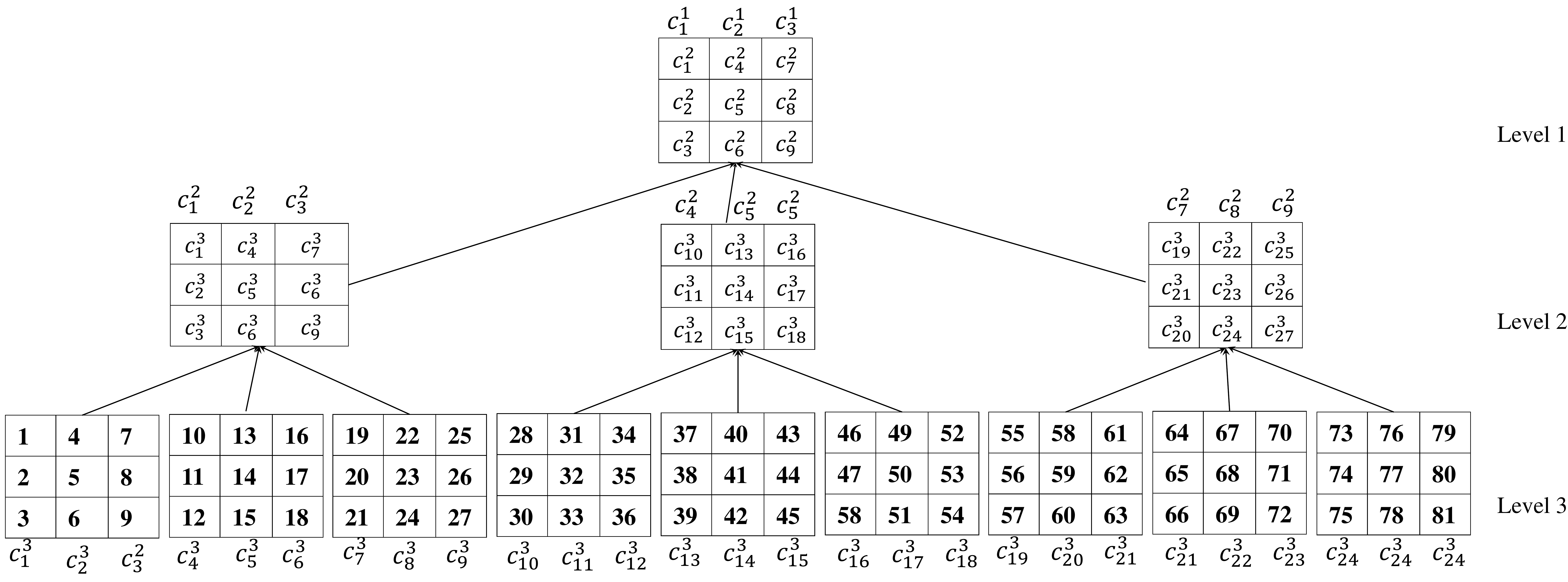}
  \end{minipage}
\caption{The \emph{second extension of the AU method} (Algorithm 3): Phase 1 -- Creation of the bottom-up tree.}
\label{fig:fig_second_ext_au_example}
\end{center}
\end{figure*}

The height of the bottom up tree is $l-1$, and the last $(l-1)^{th}$ level has $m$ inputs in the form of $\frac{m}{q^2}$ matrices of size $q\times q$. Note that we have $\frac{m}{q}$ columns at the last level, which holds $m$ inputs; and these $\frac{m}{q}$ columns are called the \emph{input columns}. We create the tree in bottom-up fashion, where $(l-2)^{th}$ level has $\frac{m}{q^3}$ matrices, whose each cell value refers to a \emph{input column} of $(l-1)^{th}$ level. We use a notation to refer a column of $i^{th}$ level by $c_j^i$, where $j$ is column index. Note that each column, $c_j^i$, at level $i$ holds $q$ columns ($c_{(j-1)q+1}^{i+1},c_{(j-1)q+2}^{i+1},\ldots c_{jq}^{i+1}$) of $(i+1)^{th}$ level. In general, there are $\frac{m}{q^{l-i+2}}$ matrices at level $i$, whose each cell value, $c_j^i$, refers to a column, $c_j^{i+1}$, of $(i+1)^{th}$ level.

Following that the bottom-up tree for $m=3^4$ identical-sized inputs and $q=3$ has height 3. The last level ($(l-1)^{th}=3^{rd}$) has $81$ inputs in the form of $\frac{m}{q^2}=9$ matrices of size $3\times 3$. Note that we have $\frac{m}{q}=24$ columns at $3^{rd}$ level; called the \emph{input columns}. The $l-2=2^{ed}$ level has $\frac{m}{q^3}=3$ matrices, whose each column, $c_j^2$, refers to $q=3$ columns ($c_{(j-1)q+1}^3,c_{(j-1)q+2}^3,\ldots c_{jq}^3$) of $3^{rd}$ level. Further, the root node is at level 1, whose each column, $c_j^1$, refers to $q=3$ columns ($c_{(j-1)q+1}^{2},c_{(j-1)q+2}^{2},\ldots c_{jq}^{2}$) of $2^{ed}$ level.
\end{example}

\medskip\noindent\textbf{\textit{The second phase: creation of an assignment tree.}} The assignment tree is created in top-down fashion. Our objective is to assign each pair of inputs to a reducer, where inputs are arranged in the \emph{input columns} of the bottom-up tree. If we can assign each pair of \emph{input columns} (of the bottom-up tree) in the form of $(q\times q)$-sized matrices, then the implementation of the \emph{AU method} on each such matrices results in an assignment of every pair of inputs to reducers. Hence, we try to make pairs of all the \emph{input column}s by creating a tree called the \emph{assignment tree}.

Here, we present a simple assignment tree for $m=3^4$ and $q=3$ (see Figure~\ref{fig:fig_second_ext_au_example_assignment_tree}).

\begin{example}[Assignment tree creation]
The root node of the bottom-up tree becomes the root node the assignment tree. Recall that the root node of the bottom-up tree is a $q\times q$ matrix. First, consider the root node to understand the working of the \emph{AU method} to create the assignment tree. Consider that each cell value of the root node matrix is of size one, and we have $(q+1)$ teams of $q$ bins (of size $q$) in each team. Our objective to use the \emph{AU method} on the root node matrix is to assign each pair of cell values ($\langle c_x^2,c_y^2\rangle$) in $q(q+1)$ bins that results in an assignment of every pair of cell values $\langle c_x^2,c_y^2\rangle$ at a bin.

\begin{figure*}[t]
\begin{center}
  \begin{minipage}{.99\textwidth}
  \centering
  \includegraphics[scale=.13]{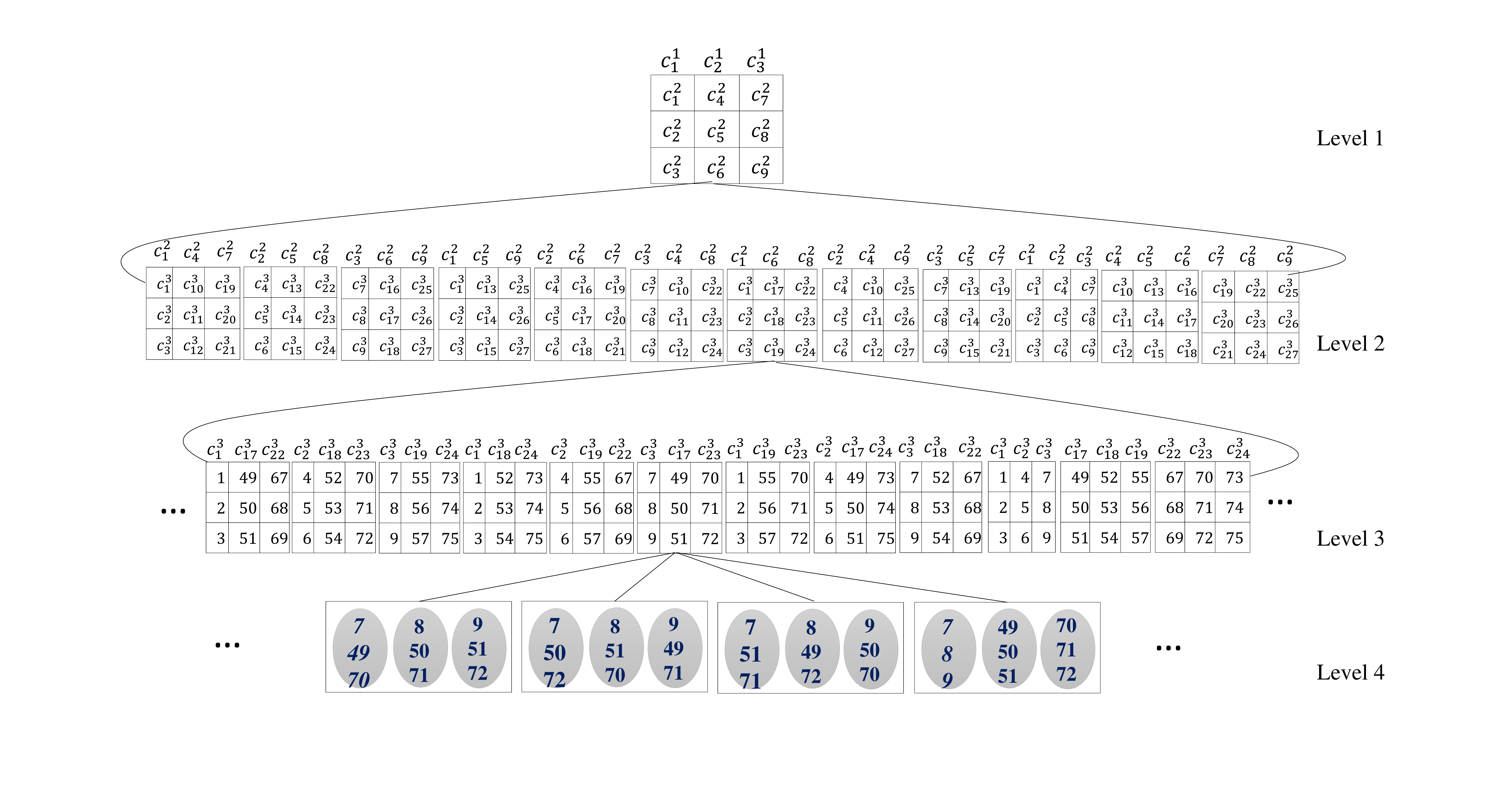}
  \end{minipage}
\caption{The \emph{second extension of the AU method} (Algorithm 3): Phase 2 -- Creation of the assignment tree.}
\label{fig:fig_second_ext_au_example_assignment_tree}
\end{center}
\end{figure*}

Now, we create matrices by using these bins (the bins created by the \emph{AU method}'s implementation on the root node) that are holding the indices of columns of the second level ($c_x^2$) of the bottom-up tree. We take each bin and its $q$ indices $c_j^2,c_{j+1}^2,\ldots c_{j+q}^2$. We replace each $c_j^2$ with $q$ columns as: $c_{(j-1)q+1}^3,c_{(j-1)q+2}^3,\ldots c_{jq}^3$ that results in $q(q+1)$ matrices of size $q\times q$, and these $q(q+1)$ matrices become child nodes of the root node. Now, we consider each such matrix separately and perform a similar operation as we did for the root node.

In this manner, the \emph{AU method} creates $(q(q+1))^{i-1}$ child nodes (that are matrices of size $q\times q$) at $i^{th}$ level of the assignment tree, and they create $(q(q+1))^i$ child nodes (matrices of size $q\times q$) at $(i+1)^{th}$ level of the assignment tree.

Recall that there are $\frac{m}{q}$ \emph{input column}s at $(l-1)^{th}$ level of the bottom-up tree that hold the original $m$ inputs. The implementation of the \emph{AU method} on each node ($q\times q$-sized) matrix of $(l-2)^{th}$ level of the assignment tree assigns each pair of \emph{input columns} at $(l-1)^{th}$ level of the assignment tree. Further the \emph{AU method}'s implementation on each matrix of $(l-1)^{th}$ level assigns every pairs of the original inputs to $q^l\times (q+1)^{l-1}$ reducers at $l^{th}$ level, which have reducers in the form of $(q(q+1))^{l-1}$ teams of $q$ reducers in each team.

For $m=3^4$ identical-sized inputs and $q=3$, we take the root node of the bottom-up tree (Figure~\ref{fig:fig_second_ext_au_example}) that becomes the root node of the assignment tree. We implement the \emph{AU method} on the root node and assign each pair of cell values ($c_j^2$, $1\leq j\leq9$) to a bin of size $q$. Each cell value of the bins ($c_j^2$) is then placed by $q=3$ columns $c_{(j-1)q+1}^3,c_{(j-1)q+2}^3,\ldots c_{jq}^3$ that results in an assignment of each pair of columns of the second level of the bottom-up tree. For clarity, we are not showing bins. For the next $3^{rd}$ level, we again implement the \emph{AU method} on all 12 matrices at $2^{nd}$ level and get 144 matrices at the third level. The matrices at $3^{rd}$ level are pairs of each \emph{input columns} (of the bottom-up tree). The \emph{AU method}'s implementation on each matrix of $3^{rd}$ level assigns each pair of original inputs to reducers. For clarity, we are only showing all the matrixes and teams at levels 3 and 4, respectively.
\end{example}

\begin{figure}
\begin{center}
\includegraphics[scale=0.35]{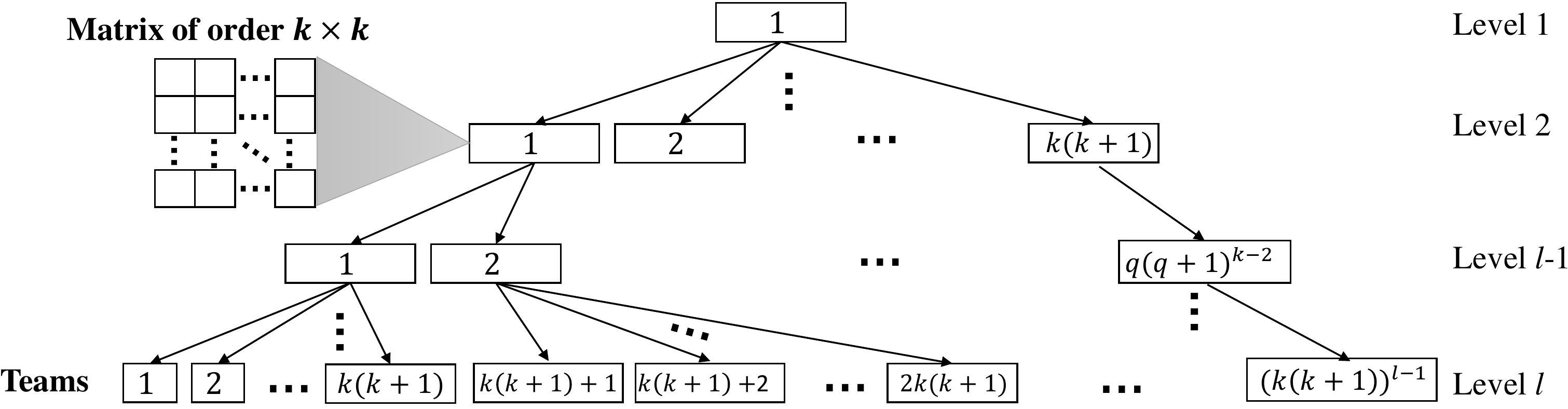}
\end{center}
\caption{An assignment tree created using Algorithm 3.}
\label{fig:fig_tree}
\end{figure}

The assignment tree uses the root node of the bottom up tree, and we implement the \emph{AU method} on the root node that results in $q(q+1)$ child nodes at level two. Each child node is a $q\times q$ matrix, and the columns of all the $q(q+1)$ matrices provide all-pairs of the cell values of the root node matrix. At level $i$, the assignment tree has $(q(q+1))^{i-1}$ nodes, see Figure~\ref{fig:fig_tree}. The height of the assignment tree is $l$, where $(l-1)^{th}$ level has all-pairs of \emph{input column}s and $l^{th}$ level has a solution to the \emph{A2A mapping schema problem} for $m$ inputs.

\medskip\medskip \noindent \emph{Algorithm correctness.} Algorithm 3 satisfies the following Lemma~\ref{lemma:assignment_tree_height}:
\begin{lemma}\label{lemma:assignment_tree_height}
The height of the assignment tree is $l$, and $l^{th}$ level of the assignment tree assigns each pairs of inputs to reducers.
\end{lemma}

\medskip\begin{theorem}[The communication cost obtained using Algorithm 3]\label{th:second ext to the au method reducer_communication_cost}
Algorithm 3 requires at most $q\times(q(q+1))^{l-1}$ reducers and results in at most $q^2\times(q(q+1))^{l-1}$ communication cost, where $q$ is the reducer capacity and $l>2$.
\end{theorem}
\begin{proof}
For a given $m=q^l$, $l>2$, the assignment tree has height $l$ (Lemma~\ref{lemma:assignment_tree_height}), and (according to Algorithm 3) $l^{th}$ level has $q\times (q(q+1))^{l-1}$ reducers providing an assignment of each pairs of inputs. Hence, Algorithm 3 uses $q(q(q+1))^{l-1}$ reducers, and the communication cost is at most $q^2\times(q(q+1))^{l-1}$.
\end{proof}

\medskip\medskip \noindent \textit{Approximation factor.} The optimal communication is $\frac{m(m-1)}{q-1}$ (see Theorem~\ref{th:a2a_2-step-The total communication cost equal input}). Replacing $m$ with $q^l$ we get $q^l(q^l-1)/(q-1)$. Thus, the ratio between the optimal communication and the communication of our mapping schema is $(q^l-1)/q(q-1)(q+1)^{l-1}$.
We can see two cases:
\begin{enumerate}
  \item When $q$ is large. Then we drop the constant 1 and the ratio is approximately equal to $\frac{1}{q}$.
  \item When $q$ is very small compared to $q^l$. Then the ratio is $q^l/q(q-1)(q+1)^{l-1}$.

  For $q=5$, the inverse of the ratio is approximately $(6/5)^{l-1}$. This is already acceptable for practical applications if we think that the size of data is $5^l$, thus $l$ may as well be $l=9$, in which case this ratio is approximately 4.3. For $q=2$ and $q=3$ we already have optimal mappings schemas. Our conjecture is that there are optimal schemas for $q=4$ and $q=5$ even by using the techniques developed and presented here.
\end{enumerate}

\medskip \noindent{\bf Open problem:} In this section, we provided two algorithms for two different cases extending the \emph{AU method}. However, this is an open problem of finding good approximation algorithms for the subcases that are not covered here.

\medskip
\section{A Hybrid Algorithm for the \textit{A2A Mapping Schema Problem}}
\label{subsec:A hybrid approach}
In the previous sections, we provide algorithms for different-sized and almost equal-sized inputs. The hybrid approach considers both different-sized and almost equal-sized inputs together. The objective of the hybrid approach is to place inputs to two different-sized bins, and then consider each of the bins as a single input.

Specifically, the hybrid approach uses the previously given algorithms (bin-packing-based approximation algorithm) and Algorithms~\ref{alg:2stepmethodforoddq}A,~\ref{alg:2stepmethodforoddq}B, 2, 3. We divide the given $m$ inputs into two disjoint sets according to their input size, and then use the bin-packing-based approximation algorithm and Algorithms~\ref{alg:2stepmethodforoddq}A,~\ref{alg:2stepmethodforoddq}B, 2, or 3 depending on the size of inputs.

\medskip\medskip \noindent \textbf{Algorithm 4.} We divide $m$ inputs into two sets $A$ that holds the input $i$ of size $\frac{q}{3}<w_i\leq\frac{q}{2}$, and $B$ holds all the inputs of sizes less than or equal to $\frac{q}{3}$. Algorithm 4 consists of four steps, as follows:

\begin{enumerate}
    \item Use the bin-packing-based approximation algorithm to place all the inputs of:
    \begin{enumerate}[noitemsep,nolistsep]
        \item the set $A$ to bins of size $\frac{q}{2}$, and each such bin is considered as a single input of size $\frac{q}{2}$ that we call the \emph{big input}. Consider that $x$ big inputs are obtained.
        \item the set $B$ twice, first to bins of size $\frac{q}{2}$, where each bin is considered as a single input of size $\frac{q}{2}$ that we call the \emph{medium input}, and second, to bins of size $\frac{q}{3}$, where each bin is also considered as a single input of size $\frac{q}{3}$ that we call the \emph{small input}. Consider that $y$ medium and $z$ small inputs are obtained.
    \end{enumerate}
    \item Use $\frac{x(x-1)}{2}$ reducers to assign each pair of big inputs.
    \item Use $x\times y$ reducers to assign each big input with each medium input.
    \item Use the \emph{AU method}, Algorithm~\ref{alg:2stepmethodforoddq}, 2, or 3 on the $z$ small inputs, depending on the case, to assign each pair of small inputs.
\end{enumerate}

\begin{figure}
\begin{center}
\includegraphics[scale=0.42]{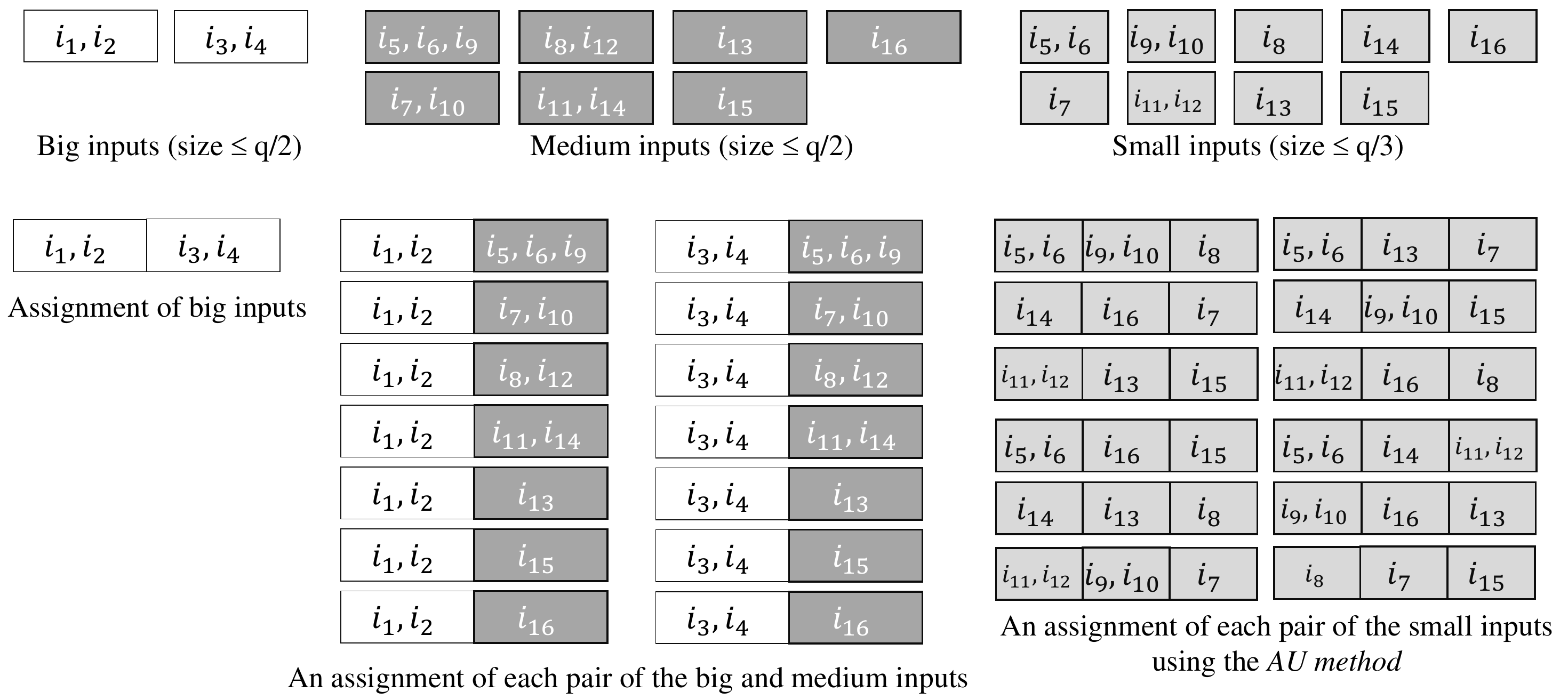}
\end{center}
\caption{An example to show the working of Algorithm 4. We are given 15 inputs, where inputs $i_1$ to $i_4$ are of sizes greater than $\frac{q}{3}$, and all the other inputs are of sizes less than or equal to $\frac{q}{3}$.}
\label{fig:hybrid_q3}
\end{figure}

We present an example to illustrate Algorithm 4 in Figure~\ref{fig:hybrid_q3}. Note that the use of $\frac{x(x-1)}{2}$ reducers assigns each pair of original inputs whose size between $\frac{q}{3}$ and $\frac{q}{2}$. Also by using $x\times y$ reducers, we assign each big input (or original inputs whose size is between $\frac{q}{3}$ and $\frac{q}{2}$) with each original input whose size is less than $\frac{q}{3}$. Further, the \emph{AU method}, Algorithm~\ref{alg:2stepmethodforoddq}, 2, or 3 assigns each pair of original inputs whose size is less than or equal to $\frac{q}{3}$.

\medskip\medskip \noindent \emph{Algorithm correctness.} The algorithm correctness shows that every pair of inputs is assigned to reducers. Specifically, the algorithm correctness shows that each pair of the big inputs is assigned to reducers, each of the big inputs is assigned to reducers with each of the medium inputs, and each pair of the small inputs is assigned to reducers.

\medskip
\section{Approximation Algorithms for the \emph{A2A Mapping Schema Problem} with an Input $>q/2$}
\label{sec:A big input of size greater than q2}
In this section, we consider the case of an input of size $w_i$, $\frac{q}{2}<w_i<q$; we call such an input as a \emph{big input}. Note that if there are two big inputs, then they cannot be assigned to a single reducer, and hence, there is no solution to the \emph{A2A mapping schema problem}. We assume $m$ inputs of different sizes are given. There is a big input and all the remaining $m-1$ inputs, which we call the \emph{small inputs}, have at most size $q-w_i$. We consider the following three cases in this section:
\begin{enumerate}
\item The big input has size $w_i$, where $\frac{q}{2}<w_i\leq \frac{2q}{3}$,
\item The big input has size $w_i$, where $\frac{2q}{3}<w_i\leq \frac{3q}{4}$,
\item The big input has size $w_i$, where $\frac{3q}{4}<w_i<q$.
\end{enumerate}

\begin{figure}
\begin{center}
\includegraphics[scale=0.4]{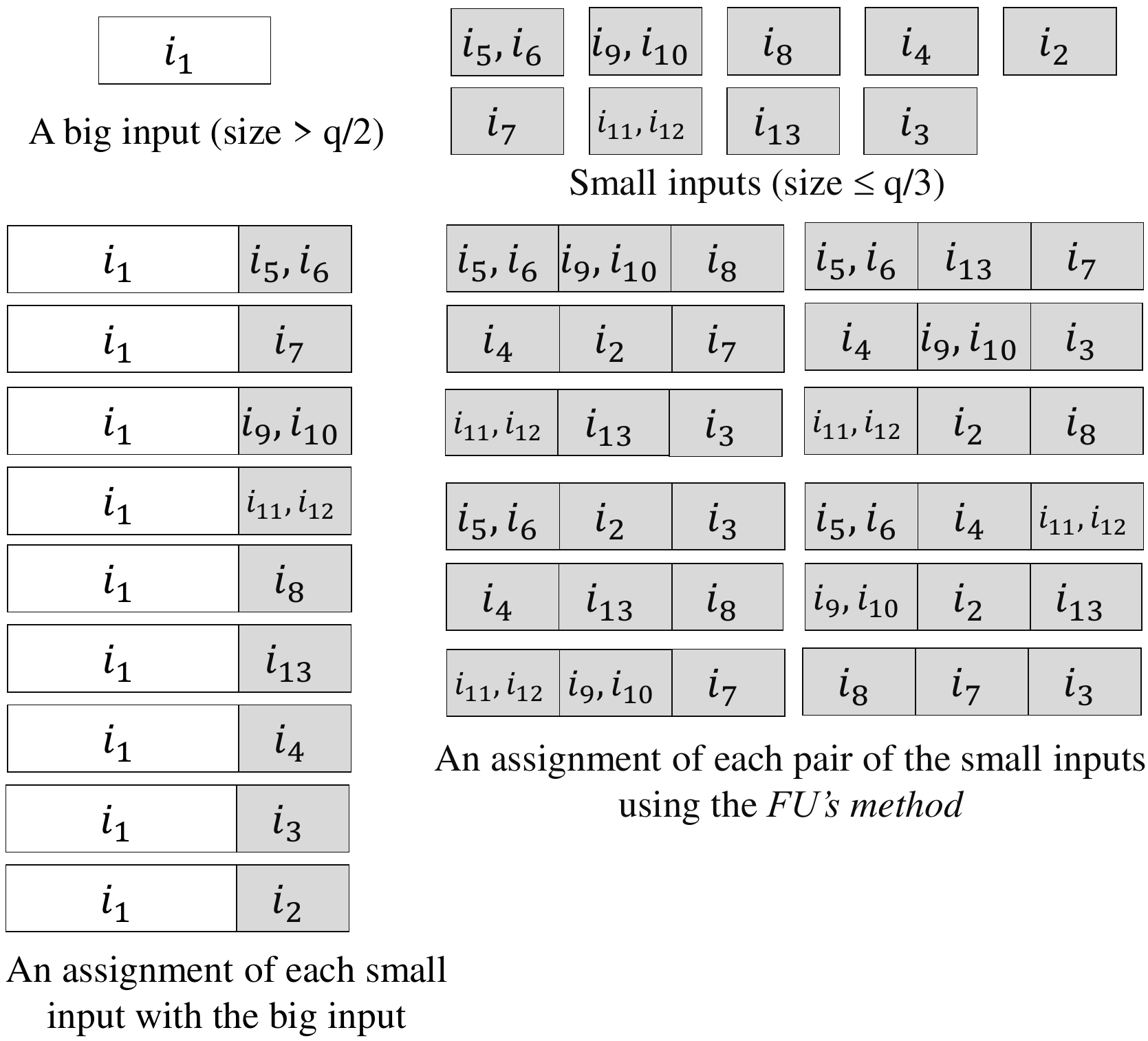}
\end{center}
\caption{An example to show an assignment of a big input of size $\frac{q}{2}< w_i\leq \frac{2q}{3}$ with all the remaining inputs of sizes less than or equal to $\frac{q}{3}$.}
\label{fig:fig_one_big_input_q3_case.pdf}
\end{figure}

The communication cost is dominated by the big input. We consider three different cases of the big input to provide efficient algorithms in terms of the communication cost, where the first two cases can assign inputs to almost an optimal number of reducers, which results in almost minimum communication cost. We use the previously given algorithms (bin-packing-based approximation algorithm) and Algorithms 1-4 to provide a solution to the \emph{A2A mapping schema problem} for the case of a big input.

A \textit{simple solution} is to use FFD or BFD bin-packing algorithm to place the small inputs to bins of size $q-w_i$. Now, we consider each of the bins as a single input of size $q-w_i$. Let $x$ bins are used. We assign each of the $x$ bins to one reducer with a copy of the big input. Further, we assign the small inputs to bins of size $\frac{q}{2}$, and consider each of such bins as a single input of size $\frac{q}{2}$. Now, we can assign each pair of bins (each of size $\frac{q}{2}$) to reducers. In this manner, each pair of inputs is assigned to reducers.

\medskip\medskip \noindent \textbf{The big input of size $\frac{q}{2}< w_i\leq \frac{2q}{3}$.} In this case, we assume that the small inputs have at most $\frac{q}{3}$ size. We use First-Fit Decreasing (FFD) or Best-Fit Decreasing (BFD) bin-packing algorithm, the \emph{AU method} (Section~\ref{subsec:AU method}), and Algorithms 2, 3 (Section~\ref{subsec:Generalizing Techniques from the AU method}). We proceed as follows:

\begin{enumerate}
  \item First assign the big input with the small inputs.
    \begin{enumerate}[noitemsep,nolistsep]
        \item Use a bin-packing algorithm to place the small inputs to bins of size $\frac{q}{3}$. Now, we consider each of the bins as a single input of size $\frac{q}{3}$.
        \item Consider that $x$ bins are used. Assign each of the bins to one reducer with a copy of the big input.
    \end{enumerate}
  \item Depending on the number of bins, we use the \emph{AU method}, and Algorithms 2, 3 to assign each pair of the small inputs to reducers.
\end{enumerate}
An example is given in Figure~\ref{fig:fig_one_big_input_q3_case.pdf}, where we place the small inputs to 9 bins of size $\frac{q}{3}$ and assign each of the bins to one reducer with a copy of the big input. Further, we implement the \emph{AU method} on 9 bins to assign each pair of the small inputs.

\medskip\medskip \noindent \textbf{The big input of size $\frac{2q}{3}< w_i\leq \frac{3q}{4}$.} In this case, we assume that the small inputs have at most $\frac{q}{4}$ size. We use a bin-packing algorithm and Algorithms~\ref{alg:2stepmethodforoddq}B (Sections~\ref{subsec:Generalizing the Technique from section q is equal to 3}). We proceed as follows:
\begin{enumerate}
  \item First assign the big input with the small inputs.
    \begin{enumerate}[noitemsep,nolistsep]
        \item Use a bin-packing algorithm to place the small inputs to bins of size $\frac{q}{4}$.
        \item Consider that $x$ bins are used. Assign each of the bins to one reducer with a copy of the big input.
    \end{enumerate}
  \item Depending on the number of bins, we use Algorithm~\ref{alg:2stepmethodforoddq}B to assign each pair of small inputs.
\end{enumerate}

\medskip\medskip \noindent \textbf{The big input of size $\frac{3q}{4}<w_i<q$.} In this case, we assume that the small inputs have at most $q-w_i$ size. In this case, we use a bin-packing algorithm and place the small inputs to bins of size $q-w_i$. We then place each of the bins to one reducer with a copy of the big input. Note that, we have not assigned each pair of small inputs. In order to assign each pair of small inputs, we use the bin-packing-based approximation algorithm (Section~\ref{subsubsec:Different sized inputs}) or Algorithms 1-4 depending on size of the small inputs.

\medskip
\begin{theorem}[Upper bounds from algorithm]\label{th:a2a_big_input}
For a list of $m$ inputs where a big input, $i$, of size $\frac{q}{2}<w_i<q$ and for a given reducer capacity $q$, $q<s^{\prime}<s$, an input is replicated to at most $m-1$ reducers for the \textit{A2A mapping schema problem}, and the number of reducers and the communication cost are at most $m-1+\frac{8{s^2}}{q^2}$ and $(m-1)q+\frac{4s^2}{q}$, respectively, where $s^{\prime}$ is the sum of all the input sizes except the size of the big input and $s$ is the sum of all the input sizes.
\end{theorem}
\begin{proof}
The big input $i$ can share a reducer with inputs whose sum of the sizes is at most $q-w_i$. In order to assign the input $i$ with all the remaining $m-1$ small inputs, it is required to assign a sublist of $m-1$ inputs whose sum of the sizes is at most $q-w_i$. If all the small inputs are of size almost $q-w_i$, then a reducer can hold the big input and one of the small inputs. Hence, the big input is required to be sent to at most $m-1$ reducers that results in at most $(m-1)q$ communication cost.

Also, each pair of all the small inputs is assigned to reducers (by first placing them to bins of size $\frac{q}{2}$ using FFD or BFD bin-packing algorithm). The assignment of all the small inputs results in at most $\frac{8{s^{\prime}}^2}{q^2}<\frac{8{s^2}}{q^2}$ reducers and at most $\frac{4{s^{\prime}}^2}{q}<\frac{4s^2}{q}$ communication cost (Theorem~\ref{th:our_bounds}). Thus, the number of reducers are at most $m-1+\frac{8{s^2}}{q^2}$ and the communication cost is at most $(m-1)q+\frac{4s^2}{q}$.
\end{proof}

\medskip\medskip \noindent \textit{Approximation factor.} The optimal communication cost (from Theorem \ref{th:a2a_communication_cost}) is $s^2/q$ and the communication cost of the algorithm (from Theorem \ref{th:a2a_big_input}) is $(m-1)q+4s^2/q$. Thus, the ratio between the optimal communication and the communication of our mapping schema is approximately $\frac{s^2}{mq^2}$.

\medskip
\section{An Approximation Algorithm for the \emph{X2Y Mapping Schema Problem}}
\label{sec:a heuristic for the X-meets-Y Mapping Schema Problem}
We propose an approximation algorithm for the \emph{X2Y mapping schema problem} that is based on bin-packing algorithms. Two lists, $X$ of $m$ inputs and $Y$ of $n$ inputs, are given. We assume that the sum of input sizes of the lists $X$, denoted by $sum_x$, and $Y$, denoted by $sum_y$, is greater than $q$. We analyze the algorithm on criteria (number of reducers and the communication cost) given in Section~\ref{sec:Approximation Algorithms Preliminary Results}. We look at the lower bounds in Theorem~\ref{th:x2y_communication_cost}, and Theorem~\ref{th:our_bounds_x2y} gives an upper bound from the algorithm. The bounds are given in Table~\ref{table:Bounds}.

\medskip\begin{theorem}\label{th:x2y_communication_cost}
\textnormal{\textsc{(Lower bounds on the communication cost and number of reducers)}} For a list $X$ of $m$ inputs, a list $Y$ of $n$ inputs, and a given reducer capacity $q$, the communication cost and the number of reducers, for the \textit{X2Y mapping schema problem}, are at least $\frac{2\cdot sum_x\cdot sum_y}{q}$ and $\frac{2\cdot sum_x\cdot sum_y}{q^2}$, respectively, where $q$ is the reducer capacity, $sum_x$ is the sum of input sizes of the list $X$, and $sum_y$ is the sum of input sizes of the list $Y$.
\end{theorem}
\begin{proof}
Since an input $i$ of the list $X$ and an input $j$ of the list $Y$ are replicated to at least $\frac{sum_y}{q}$ and $\frac{sum_x}{q}$ reducers, respectively, the communication cost for the inputs $i$ and $j$ are $w_i\times \frac{sum_y}{q}$ and $w_j\times \frac{sum_x}{q}$, respectively. Hence, the communication cost will be at least $\sum_{i=1}^m w_i \frac{sum_y}{q}+\sum_{j=1}^n w_j \frac{sum_x}{q} = \frac{2\cdot sum_x\cdot sum_y}{q}$.

Since the number of bits to be assigned to reducers is at least $\frac{2\cdot sum_x\cdot sum_y}{q}$ and a reducer can hold inputs whose sum of the sizes is at most $q$, the number of reducers must be at least $\frac{2\cdot sum_x\cdot sum_y}{q^2}$.
\end{proof}

\medskip\medskip \noindent \textbf{Bin-packing-based approximation algorithm for the \emph{X2Y mapping schema problem.}} A solution to the \emph{X2Y mapping schema problem} for different-sized inputs can be achieved using bin-packing algorithms. Let two lists $X$ of $m$ inputs and $Y$ of $n$ inputs are given. The algorithm will not work when a list holds an input of size $w_i$ and the another list holds an input of size greater than $q-w_i$, because these inputs cannot be assigned to a single reducer in common. Let the size of the largest input, $i$, of the list $X$ is $w_i$; hence, all the inputs of the list $Y$ have at most size $q-w_i$. We place inputs of the list $X$ to bins of size $w_i$, and let $x$ bins are used to place $m$ inputs. Also, we place inputs of the list $Y$ to bins of size $q-w_i$, and let $y$ bins are used to place $n$ inputs. Now, we consider each of the bins as a single input, and a solution to the \emph{X2Y mapping schema problem} is obtained by assigning each of the $x$ bins with each of the $y$ bins to reducers. In this manner, we require $x\cdot y$ reducers.

\medskip\begin{theorem}[Upper bounds from the algorithm]\label{th:our_bounds_x2y}
For a bin size $b$, a given reducer capacity $q=2b$, and with each input of lists $X$ and $Y$ being of size at most $b$, the number of reducers and the communication cost, for the \textit{X2Y mapping schema problem}, are at most $\frac{4\cdot sum_x \cdot sum_y}{b^2}$, and at most $\frac{4\cdot sum_x \cdot sum_y}{b}$, respectively, where $sum_x$ is the sum of input sizes of the list $X$, and $sum_y$ is the sum of input sizes of the list $Y$.
\end{theorem}
\begin{proof}
A bin $i$ can hold inputs whose sum of the sizes is at most $b$. Hence, it is required to divide inputs of the lists $X$ and $Y$ into at least $\frac{sum_x}{b}$ and $\frac{sum_y}{b}$ bins, respectively. Since the FFD or BFD bin-packing algorithm ensures that all the bins (except only one bin) are at least half-full, each bin of size $b$ has at least inputs whose sum of the sizes is at least $\frac{b}{2}$. Thus, all the inputs of the lists $X$ and $Y$ can be placed in at most $\frac{sum_x}{b/2}$ and $\frac{sum_y}{b/2}$ bins of size $b$, respectively.

Let $x$ (=$\frac{2\cdot sum_x}{b}$) and $y$ (=$\frac{2\cdot sum_y}{b}$) bins are used to place inputs of the lists $X$ and $Y$, respectively. Since each bin is considered as a single input, we can assign each of the $x$ bins with each of the $y$ bins at reducers, and hence, we require at most $\frac{4\cdot sum_x \cdot sum_y}{b^2}$ reducers. Since each bin that is containing inputs of the list $X$ (resp. $Y$) is replicated to at most $\frac{2\cdot sum_y}{b}$ (resp. at most $\frac{2\cdot sum_x}{b}$) reducers, the replication of individual inputs of the list $X$ (resp. $Y$) is at most $\frac{2\cdot sum_y}{b}$ (resp. at most $\frac{2\cdot sum_x}{b}$) and the communication cost is at most $\sum_{1\leq i\leq m} w_i \times \frac{2\cdot sum_y}{b} + \sum_{1\leq j\leq n} w_j \times \frac{2\cdot sum_x}{b} = \frac{4\cdot sum_x \cdot sum_y}{b}$.
\end{proof}

\medskip\medskip \noindent \textit{Approximation factor.} The optimal communication is $\frac{2\cdot sum_x\cdot sum_y}{q}$. Thus, the ratio between the optimal communication and the communication of our mapping schema is $\frac{1}{4}$.

\medskip
\section{Conclusion}
\label{sec:conclusion}
Two new important practical aspects in the context of MapReduce, namely different-sized inputs and the reducer capacity, are introduced for the first time. The capacity of a reducer is defined in terms of the reducer's memory size. We note that processing time is typically proportional to the memory capacity. All reducers have an identical capacity, and any reducer cannot hold inputs whose input sizes are more than the reducer capacity. We demonstrated the importance of the capacity aspect by considering two common mapping schema problems of MapReduce, \emph{A2A mapping schema problem} -- every two inputs are required to be assigned to at least one common reducer -- \emph{X2Y mapping schema problem} -- every two inputs, the first input from a list $X$ and the second input from a list $Y$ -- is required to be assigned to at least one common reducer. Unfortunately, it turned out that finding solutions to the \emph{A2A} and the \emph{X2Y} mapping schema problems that use the minimum number of reducers is not possible in polynomial time. On the positive side, we present near optimal approximation algorithms for the \emph{A2A} and the \emph{X2Y} mapping schema problems.

Mapping schemes for the case of reducers with different capacities are left for future research. Nevertheless, there exist a reduction to our proposed algorithms that may yield a reasonable performance in some cases. In particular, we can consider a common divisor of all the non-identical reducer capacity as a unit-sized reducer capacity. Then, we can follow our proposed algorithms to solve problems while regarding non-identical reducer capacities.

\appendix
\medskip
\section{Proofs of Theorems 1, 2, and Lemma 1}\label{app:proof of th}
\medskip\noindent \textbf{Theorem 1} \textit{The problem of finding whether a mapping schema of $m$ inputs of different input sizes exists, where every two inputs are assigned to at least one of $z\geq 3$ identical-capacity reducers, is NP-hard.}

\smallskip
\begin{proof}
The proof is by a reduction from the partition problem~\cite{DBLP:books/fm/GareyJ79} that is a known NP-complete problem. The partition problem is defined as follows: given a set $I = \{i_1, i_2, \ldots, i_m\}$ of $m$ positive integer numbers, it is required to find two disjoint subsets, $S_1 \subset I$ and $S_2 \subset I$, so that the sum of numbers in $S_1$ is equal to the sum of numbers in $S_2$, $S_1 \cap S_2 = \emptyset$, and $S_1 \cup S_2 =I$.

We are given $m$ inputs whose input size list is $W=\{w_1, w_2, \ldots, w_m\}$, and the sum of the sizes is $s= \Sigma_{1\leq i \leq m}w_i$. We add $z-3$ additional inputs, $ai_1, ai_2, \ldots, ai_{z-3}$, each of size $\frac{s}{2}$. We call these new $z-3$ ($ai_1, ai_2, \ldots, ai_{z-3}$) inputs the \emph{medium inputs}. In addition, we add one more additional input, $ai^{\prime}$, of size $\frac{(z-2)s}{2}$ that we call the \emph{big input}. Further, we assume that the reducer capacity is $\frac{(z-1)s}{2}$.

The proof proceeds in two steps: (\textit{i}) we prove that in case the $m$ original inputs can be partitioned, then all the inputs can be assigned to the $z$ reducers such that every two inputs are assigned to at least one reducer, (\textit{ii}) we prove that in case the mapping schema for all the inputs over the $z$ reducers is successful, then there are two disjoint subsets $S_1$ and $S_2$ of the $m$ original inputs that satisfy the partition requirements. We can assume that if the sum is not divisible by 2, then the answer to the partition problem is surely \enquote{no,} so the reduction of the partition problem to the \emph{A2A mapping schema problem} is trivial.

\begin{figure}[h!]
\begin{center}
\includegraphics{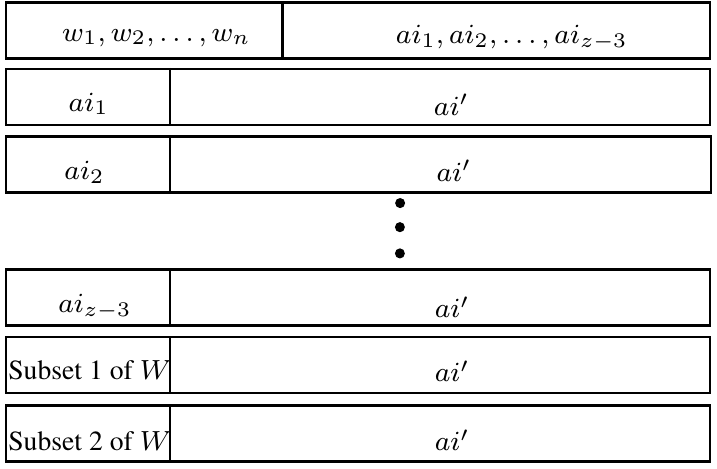}
\end{center}
\caption{Proof of NP-hardness of the \emph{A2A mapping schema problem} for $z>2$ identical-capacity reducers, Theorem~\ref{th:all-to-all}.}
\label{fig:new_proof}
\end{figure}

We first show that if there are two disjoint subsets $S_1$ and $S_2$ of equal size of the $m$ original inputs, then there must exist a solution to the \emph{A2A mapping schema problem}. Recall that any of the reducers can hold a set of inputs whose sum of the sizes is at most $\frac{(z-1)s}{2}$, and the sum of the sizes of the new $z-3$ medium inputs is exactly $\frac{(z-3)s}{2}$. Hence, all the $m$ original inputs ($i_1, i_2, \ldots, i_m$) and a list of the $z-3$ medium inputs can be assigned to a single reducer (out of the $z$ reducers), and this assignment uses $s+\frac{(z-3)s}{2}$ capacity, which is exactly the capacity of any reducer. Further, the big input, $ai^{\prime}$, of size $\frac{(z-2)s}{2}$ can share the same reducer with only one medium input $ai_i$ (it could also share with original inputs). Thus, the big input, $ai^{\prime}$, and all the medium inputs are assigned to $z-3$ reducers (out of the remaining $z-1$ reducers). In addition, the remaining two reducers can be used for the following assignment: the first reducer is assigned the set $S_1$ and the big input, $ai^{\prime}$, and the second reducer is assigned the set $S_2$ and the big input, $ai^{\prime}$. The above assignment is a solution to the \emph{A2A mapping schema problem} for the given $m$ original inputs, the $z-3$ medium inputs, and the big input using $z$ reducers, see Figure~\ref{fig:new_proof}.

Now, we show that a solution to the \emph{A2A mapping schema problem} --- for all the inputs over the $z$ reducers --- results in a partition of the $m$ original inputs into two equal-sized blocks. We also show that in a solution to the \emph{A2A mapping schema problem}, each of the $m$ original inputs and every medium input, $ai_i$, are assigned to exactly two reducers, and the big input, $ai^{\prime}$, is assigned to exactly $z-1$ reducers. Recall that the total sum of the sizes is $s+\frac{(z-3)s}{2}+\frac{(z-2)s}{2} = \frac{(2z-3)s}{2}$.

Due to the reducer capacity of a single reducer, all the inputs cannot be assigned to a single reducer; only a subset of the inputs, whose sum of the sizes is at most $\frac{(z-1)s}{2}$, can be assigned to one reducer. Thus, each input is assigned to at least two reducers in order to be coupled with all the other inputs.

Moreover, the big input, $ai^{\prime}$, can share the same single reducer with only a sublist, $S^{\prime}$, whose sum of the sizes is at most $\frac{s}{2}$. Hence, the big input, $ai^{\prime}$, is required to be assigned to at least $z-3$ reducers in order to be paired with the medium inputs $ai_i$. Furthermore, the big input, $ai^{\prime}$, can share the same reducer with a sublist of the $m$ original inputs whose sum of the sizes is at most $\frac{s}{2}$. This fact means that the big input, $ai^{\prime}$, must be assigned to two more reducers. On the other hand, all the medium inputs can share the same reducer with the original $m$ inputs. Thus, here, the total reducer capacity occupied by all the inputs is $2 \times \Sigma_{1 \leq i \leq m} \: w_i + 2 \times \frac{(z-3)s} {2} + (z-1)\times \frac{(z-2)s}{2} = 2s+(z-3)s+\frac{(z-1)(z-2)s}{2} = \frac{(z-1)zs}{2}$, which is exactly the total capacity of all the $z$ reducers. Thus, each of the $m$ original inputs and each medium input $ai_i$ cannot be assigned more than twice, and hence, each is assigned exactly twice. In addition, the big input, $ai^{\prime}$, is assigned to exactly $z-1$ reducers. This fact also shows that all the reducers are entirely filled with distinct inputs. Thus, a solution to the \emph{A2A mapping schema problem} yields partitions of the $m$ original inputs to $S_1$ and $S_2$ blocks, where the sum of the input sizes of any block is exactly $\frac{s}{2}$. Therefore, if there is a polynomial-time algorithm to construct the mapping schema, where every input is required to be paired with every other input, then the mapping schema finds the partitions of the $m$ original inputs in polynomial time.
\end{proof}

\medskip
\noindent\textbf{Theorem 2} \textit{The problem of finding whether a mapping schema of $m$ and $n$ inputs of different input sizes that belongs to list $X$ and list $Y$, respectively, exists, where every two inputs, the first from $X$ and the second from $Y$, are assigned to at least one of $z\geq2$ identical-capacity reducers, is NP-hard.}

\smallskip\begin{proof}
The proof is by a reduction from the partition problem~\cite{DBLP:books/fm/GareyJ79} that is a known NP-complete problem. We are given a list of inputs $I=\{i_1,i_2,\ldots,i_m\}$ whose input size list is $W=\{w_1, w_2, \ldots, w_m\}$, and the sum of the sizes is $s= \Sigma_{1\leq i \leq m} w_i$. We add $z-2$ additional inputs, $ai_1, ai_2, \ldots, ai_{z-2}$, each of size $\frac{s}{2}$. We call these new $z-2$ ($ai_1, ai_2, \ldots, ai_{z-2}$) inputs the \emph{big inputs}. In addition, we add one more additional input, $ai^{\prime}$, of size 1 that we call the \emph{small input}. Further, we assume that the reducer capacity is $1+\frac{s}{2}$. Now, the list $I$ holds $m+z-1$ inputs.

For the \emph{X2Y mapping schema problem}, we consider $m$ original inputs and the $z-2$ big inputs as a list $X$, and the small input as a list $Y$. A solution to the \emph{X2Y mapping schema problem} assigns each of the $m$ original inputs and each big input (of the list $X$) with the small input of the list $Y$.

\begin{figure}[h!]
\begin{center}
\includegraphics{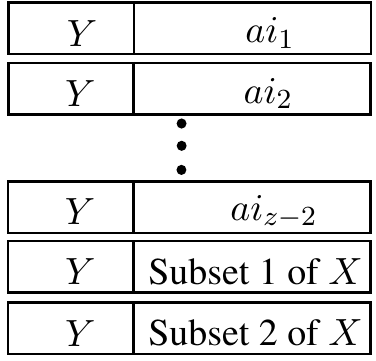}
\end{center}
\caption{Proof of NP-hardness of the \emph{X2Y mapping schema problem} for $z>1$ identical-capacity reducers, Theorem~\ref{th:x-meets-y}.}
\label{fig:x-meets-y}
\end{figure}

The proof proceeds in two steps: (\textit{i}) we prove that in case the $m$ original inputs can be partitioned, then all the $m$ original inputs, the $z-2$ big inputs, and the small input can be assigned to the $z$ reducers such that they satisfy the \emph{X2Y mapping schema problem}, (\textit{ii}) in case the \emph{X2Y mapping schema problem} is successful, then there are two disjoint subsets, $S_1$ and $S_2$, of the $m$ original inputs that satisfy the partition requirements.

We first show that if there are two disjoint subsets $S_1$ and $S_2$ of equal size of the $m$ original inputs, then there must exist a solution to the \emph{X2Y mapping schema problem}. Recall that any of the reducers can hold a set of inputs whose sum of sizes is at most $1+\frac{s}{2}$, and the sum of the sizes of the new $z-2$ big inputs is exactly $\frac{s}{2}$. Hence, the small input, $ai^{\prime}$, of size 1 and each big input, $ai_i$, can be assigned to $z-2$ reducers (out of the $z$ reducers), and this assignment uses $1+\frac{s}{2}$ capacity, which is exactly the capacity of any reducer. In addition, the remaining two reducers can be used for the following assignment: the first remaining reducer is assigned the set $S_1$ and the small input, $ai^{\prime}$, and the second remaining reducer is assigned the remaining original inputs, $S_2$, and the small input, $ai^{\prime}$. The above assignment is a solution to the \emph{X2Y mapping schema problem} (for the given $m+z-2$ inputs of the list $X$ and the one input of the list $Y$ using $z$ reducers, see Figure~\ref{fig:x-meets-y}).

Now, we prove the second claim that a solution to the \emph{X2Y mapping schema problem} results in a partition of the $m$ original inputs into two equal-sized blocks. Recall that the total sum of the sizes is $s+\frac{(z-2)s}{2}+1 = \frac{z\times s}{2}+1$.

Due to the reducer capacity of a single reducer, all the inputs cannot be assigned to a single reducer; only a sublist of the inputs, whose sum of the sizes is at most $1+\frac{s}{2}$, can be assigned to a single reducer. We show that the small input, $ai^{\prime}$, must be assigned to all the $z$ reducers. The small input, $ai^{\prime}$, of size one can share the same single reducer with only a subset, $S^{\prime}$, whose sum of the sizes is at most $\frac{s}{2}$. Hence, the small input, $ai^{\prime}$, is required to be assigned to $z-2$ reducers (out of $z$ reducers) in order to be paired with all the big inputs $ai_i$. and the remaining two reducers in order to be paired with all the $m$ original inputs. This fact results in that a solution to the \emph{X2Y mapping schema problem} yields partitions of the $m$ original inputs to $S_1$ and $S_2$ blocks, where the sum of the input sizes of any block is exactly $\frac{s}{2}$. Therefore, if there is a polynomial-time algorithm to construct the mapping schema, where every input of one list is required to be paired with every other input of another list, then the mapping schema finds the partitions of the $m$ original inputs in polynomial time.
\end{proof}

\medskip
\noindent\textbf{Lemma 1} \textit{Let $q$ be the reducer capacity. Let the size of an input is $\big\lceil\frac{q-1}{2}\big\rceil$. Each pair of $u=2^i$, $i>0$, inputs can be assigned to $2^i - 1$ teams of $2^{i-1}$ reducers in each team.}
\smallskip
\begin{proof}
The proof is by induction on $i$.

\noindent \textbf{Basis case.} For $i=1$, we have $u=2$ inputs, and we can assign them to a team of one reducer of capacity $q$. Hence, Lemma~\ref{lm:recursive_def_odd} holds for ($i=1$) two inputs.

\noindent \textbf{Inductive step.} Assume that the inductive hypothesis --- there is a solution for $u=2^{i-1}$ inputs, where all-pairs of $u=2^{i-1}$ inputs are assigned to $2^{i-1} - 1$ teams of $2^{i-2}$ reducers in each team and have the team property (each team has one occurrence of each input, which we will prove in algorithm correctness) --- is true. Now, we can build a solution for $u=2^i$ inputs, as follows:

\begin{enumerate}[label=(\alph*),noitemsep,nolistsep]
  \item Divide $u=2^i$ inputs into two groups of $2^{i-1}$ inputs in each group,
  \item Recursively create teams for each of the two groups,
  \item Create some of the teams for the $2^i$ inputs by combining the $j^{th}$ team from the first group with the $j^{th}$ team from the second group. Since by the inductive hypothesis we have a solution for $u=2^{i-1}$ inputs, we can assign inputs of these two groups to $2\cdot(2^{i-1} - 1)$ teams of $2^{i-2}$ reducers in each team. And, by combining $j^{th}$, where $j = 1, 2, \ldots, (2^{i-1} - 1)$, teams of each group, there are $2^{i-1} - 1$ teams of $2^{i-1}$ reducers in each team; see Teams 5-7 for 8 inputs in Figure~\ref{fig:The 2-step method for q 3}.
  \item Create $2^{i-1}$ additional teams that pair the inputs from the first group with inputs from the second group. In each team, the $j^{th}$ input from the first group is assigned to the $j^{th}$ reducer. In the first team, the $j^{th}$ input from the second group is also assigned to the $j^{th}$ reducer. In subsequent teams, the assignments from the second group rotate, so in the $t^{th}$ team, the $j^{th}$ input from the second group is assigned to reducer $k+j-(2^{i-1}-1) (\mathit{modulo} 2^{i-1})$; see Teams 1-4 for 8 inputs in Figure~\ref{fig:The 2-step method for q 3}.
\end{enumerate}
By steps (c) and (d), there are total $2^{i-1} - 1+2^{i-1}=2^i-1$ teams of $2^{i-1}$ reducers in each team, and these teams holds each pair of the $u=2^i$ inputs.
\end{proof}

\medskip
\section{Pseudocode and Correctness of Algorithm 1A}
\label{app:algo_correctness_1}

\LinesNotNumbered \begin{algorithm*}[t]
\textbf{Inputs:} $m$: the number of bins obtained after placing all the given $m^{\prime}$ inputs (of size $\leq \frac{q}{k}$, $k>3$ is an odd number) to bins each of size $\frac{q}{k}$,

 $q$: the reducer capacity.

{\bf Variables:}

$A$: A set $A$, where the total inputs in the set $A$ is $y = \big\lfloor \frac{q}{2} \big\rfloor (\big\lfloor \frac{2m}{q+1} \big\rfloor+1) $

$B$: A set $B$, where the total inputs in the $B$ is $x=m-y$

$\mathit{Team}[i,j]:$ represents teams of reducers, where index $i$ indicates $i^{th}$ team and index $j$ indicates $j^{th}$ reducer in $i^{th}$ team. Consider $u = \big\lceil\frac{y}{q-\lceil q/2\rceil}\big\rceil$. There are $u-1$ teams of $v=\big\lceil \frac{u}{2}\big\rceil$ reducers in each team.

$groupA[]:$ represents disjoint groups of inputs of the set $A$, where $groupA[i]$ indicates $i^{th}$ group of $\big\lceil\frac{q-1}{2}\big\rceil$ inputs of the set $A$.

\nl{\bf Function $create\_group(y)$} \nllabel{ln:function_create_group}
\Begin{

\nl \lFor{$i\leftarrow 1$ \KwTo $u$}{$groupA[i]\leftarrow \langle i, i+1\ldots, i+\frac{q-1}{2}-1\rangle, i\leftarrow i+ \frac{q-1}{2}$ \nllabel{ln:make_group}}

\nl $2\_step\_odd\_q(1, u)$, $\mathit{Last\_Team(groupA[])}$, $Assign\_input\_from\_B(Team[])$ \nllabel{ln:call_43-method}}


\nl{\bf Function $2\_step\_odd\_q(lower, upper)$} \nllabel{ln:function_2step_method}
\Begin{

\nl \lIf{$\big\lfloor \frac{upper-lower}{2}\big\rfloor < 1$}{\KwRet}

\nl \Else{

\nl $mid\leftarrow \big\lceil\frac{upper-lower}{2}\big\rceil$, $\mathit{Assignment(lower, mid, upper)}$ \nllabel{ln:midvalue}


\nl $2\_step\_odd\_q(lower, mid)$, $2\_step\_odd\_q(mid+1, upper)$ \nllabel{ln:first_second_divide}}}

\nl \textbf{Function $\mathit{Assignment(lower, mid, upper)}$} \nllabel{ln:assignment_function_body}
\Begin{

\nl \While{$mid >1$}{

\nl \lForEach {$(a, t) \in [lower,lower+mid-1] \times [0,mid-1] $} {

$\mathit{Team}\big[(u-2\cdot mid+1)+t,a-\big\lfloor \frac{a-1}{mid}\big\rfloor\cdot\frac{mid}{2}\big]\leftarrow \langle groupA[a], groupA[value\_b(a, t, mid,upper)]\rangle$ \nllabel{ln:assign_each_pair_of_derived_input}}}}

\nl \textbf{Function $value\_b(a, t,mid, upper)$}
\Begin{

\nl \lIf{$a+t+mid <upper +1$}{\KwRet($a+t+mid$) \nllabel{ln:b_value_for_upper_plus_one}}

\nl \lElseIf{$a+t+mid >upper$}{\KwRet($a+t$) \nllabel{ln:b_value_for_greater_than_upper}}}

\nl \textbf{Function $\mathit{Last\_Team(lower, mid, upper)}$} \nllabel{ln:function_last_team}
\Begin{

\nl \lForEach {$i\in [1,v]$}{$\mathit{Team}[u-1,i]\leftarrow groupA[2\times i-1],groupA[2\times i]$ \nllabel{ln:assign_last_team}}}

\nl \textbf{Function $Assign\_input\_from\_B(Team[])$} \nllabel{ln:function_assign_b}
\Begin{

\nl \lForEach{$(i,j)\in[1,u-1]\times[1,v]$}{$\mathit{Team}[i,j]\leftarrow B[i]$ \nllabel{ln:assign_b}}}
\caption{\textbf{Part A}}
\label{alg:2stepmethodforoddq}
\end{algorithm*}

\medskip\noindent \textit{Algorithm~\ref{alg:2stepmethodforoddq}A description.} First, we divide $m$ inputs (that are actually bins of size $\frac{q}{k}$, $k>3$, after placing all the given $m$ inputs to $m^{\prime}$ bins, each of size $\frac{q}{k}$) into two sets $A$ of $y$ inputs and $B$ of $x$ inputs. Then, we make $u=\big\lceil\frac{y}{q-\lceil q/2\rceil}\big\rceil$ disjoint groups of $y$ inputs of the set $A$ such that each group holds $\frac{q-1}{2}$ inputs, lines~\ref{ln:function_create_group},~\ref{ln:make_group}. (Now, each of the groups is considered as a single input that we call the \textit{derived input}.) We do not show the addition of dummy inputs and assume that $u$ is a power of 2. Function $2\_step\_odd\_q(lower, upper)$ recursively divides the derived inputs into two halves, line~\ref{ln:function_2step_method}. Function $\mathit{Assignment(lower, mid, upper)}$ (line~\ref{ln:assignment_function_body}) pairs every two derived inputs and assigns them to the respective reducers (line~\ref{ln:assign_each_pair_of_derived_input}). Each reducer of the last team is assigned using function $\mathit{Last\_Team(groupA[])}$, lines~\ref{ln:function_last_team},~\ref{ln:assign_last_team}.

Note that functions $2\_step\_odd\_q(lower, upper)$, $\mathit{Assignment(lower, mid, upper)}$, and $value\_b(lower, t, mid,upper)$ take two common parameters, namely $lower$ and $upper$ where $lower$ is the first derived input and $upper$ is the last derived input (\textit{i}.\textit{e}., $u^{th}$ group) at the time of the first call to functions, line~\ref{ln:call_43-method}. Once all-pairs of the derived inputs are assigned to reducers, line~\ref{ln:assign_each_pair_of_derived_input}, function $Assign\_input\_from\_B(Team[])$ assigns $i^{th}$ input of the set $B$ to all the $\big\lceil\frac{u}{2}\big\rceil$ reducers of $i^{th}$ team, lines~\ref{ln:function_assign_b},~\ref{ln:assign_b}. After that, Algorithm~\ref{alg:2stepmethodforoddq}A is invoked over inputs of the set $B$ to assign each pair of the remaining inputs of the set $B$ to reducers until every pair to the remaining inputs is assigned to reducers.

The algorithm correctness proves that every pair of inputs is assigned to reducers. Specifically, we prove that all those pairs of inputs, $\langle i,j\rangle$ and $\langle i^{\prime},j^{\prime}\rangle$, of the set $A$ are assigned to a team whose $i\neq i^{\prime}$ and $j\neq j^{\prime}$ (Claim~\ref{claim:algo_odd_q_claim_1}). Then that all the inputs of the set $A$ appear exactly once in each team (Claim~\ref{claim:algo_odd_q_claim_2}). We then prove that the set $B$ holds $x\leq y-1$ inputs, when $q=3$ (Claim~\ref{claim:algo_odd_q_claim_3}). At last we conclude in Theorem~\ref{th:algo_odd_q_theorem} that Algorithm~\ref{alg:2stepmethodforoddq}A assigns each pair of inputs to reducers.

Note that we are proving all the above mentioned claims for $q=3$; the cases for $q>3$ can be generalized trivially where we make $u=\big\lceil \frac{y}{q-\lceil q/2\rceil}\big\rceil$ derived inputs from $y$ inputs of the set $A$ (and assign in a manner that all the inputs of the $A$ are paired with all the remaining $m-1$ inputs).

\medskip\begin{claim}\label{claim:algo_odd_q_claim_1}
Pairs of inputs $\langle i, j\rangle$ and $\langle i^{\prime}, j^{\prime}\rangle$, where $i=i^{\prime}$ or $j=j^{\prime}$, of the set $A$ are assigned to different teams.
\end{claim}
\begin{proof}
First, consider $i=i^{\prime}$ and $j\neq j^{\prime}$, where $\langle i, j\rangle$ and $\langle i^{\prime}, j^{\prime}\rangle$ must be assigned to two different teams. If $j\neq j^{\prime}$, then both the $j$ values may have an identical value of $\mathit{lower}$ and $\mathit{mid}$ but they must have two different values of $t$ (see lines~\ref{ln:b_value_for_upper_plus_one},~\ref{ln:b_value_for_greater_than_upper} of Algorithm~\ref{alg:2stepmethodforoddq}A), where $j=lower+t+mid$ or $j=lower+t$. Thus, for two different values of $j$ , we use two different values of $t$, say $t_1$ and $t_2$, that results in an assignment of $\langle i,j\rangle$ and $\langle i^{\prime},j^{\prime}\rangle$ to two different teams $t_1$ and $t_2$, (note that teams are also selected based on the value of $t$, $(y-2\cdot mid+1)+t$, see line~\ref{ln:assign_each_pair_of_derived_input} of Algorithm~\ref{alg:2stepmethodforoddq}A, where for $q=3$, we have $u=y$).
Suppose now that $i\neq i^{\prime}$ and $j=j^{\prime}$, where $\langle i, j\rangle$ and $\langle i^{\prime}, j^{\prime}\rangle$ must be assigned to two different teams. In this case, we also have two different values of $t$, and hence, two different $t$ values assign $\langle i,j\rangle$ and $\langle i^{\prime}, j^{\prime}\rangle$ to two different teams ($(y-2\cdot mid+1)+t$, line~\ref{ln:assign_each_pair_of_derived_input} of Algorithm~\ref{alg:2stepmethodforoddq}A).

Hence, it is clear that pairs $\langle i, j\rangle$ and $\langle i^{\prime}, j^{\prime}\rangle$, where $i\neq i^{\prime}$ and $j\neq j^{\prime}$, are assigned to a team.
\end{proof}

\medskip\begin{claim}\label{claim:algo_odd_q_claim_2}
All the inputs of the set $A$ appear exactly once in each team.
\end{claim}
\begin{proof}
There are the same number of pairs of inputs of the set $A$ and the number of reducers ($(y-1)\big\lceil\frac{y}{2}\big\rceil$) that can provide a solution to the \emph{A2A mapping schema problem} for the $y$ inputs of the set $A$. Recall that $(y-1)\big\lceil\frac{y}{2}\big\rceil$ reducers are arranged in the form of $(y-1)$ teams of $\big\lceil\frac{y}{2}\big\rceil$ reducers in each team, when $q=3$. Note that if there is a input pair $\langle i,j\rangle$ in team $t$, then the team $t$ cannot hold any pair that has either $i$ or $j$ in the remaining $\big\lceil\frac{y}{2}\big\rceil-1$ reducers. For the given $y$ inputs of the set $A$, there are at most $\big\lceil\frac{y}{2}\big\rceil$ disjoint pairs $\langle i_1, j_1\rangle$, $\langle i_2, j_2\rangle$, $\ldots$, $\langle i_{\lceil y/2\rceil}, j_{\lceil y/2\rceil}\rangle$ such that $i_1\neq i_2\neq \ldots\neq i_{\lceil y/2\rceil}\neq j_1\neq j_2\neq \ldots\neq j_{\lceil y/2\rceil}$. Hence, all $y$ inputs of the set $A$ are assigned to a team, where no input is assigned twice in a team.
\end{proof}

\medskip\begin{claim}\label{claim:algo_odd_q_claim_3}
When the reducer capacity $q=3$, the set $B$ holds at most $x\leq y-1$ inputs.
\end{claim}
\begin{proof}
Since a pair of inputs of the set $A$ requires at most $q-1$ capacity of a reducer and each team holds all the inputs of the set $A$, an input from the set $B$ can be assigned to all the reducers of the team. In this manner, all the inputs of the set $A$ are also paired with an input of the set $B$. Since there are $y-1$ teams and each team is assigned an input of the set $B$, the set $B$ can hold at most $x\leq y-1$ inputs.
\end{proof}

\medskip\begin{theorem}\label{th:algo_odd_q_theorem}
Algorithm~\ref{alg:2stepmethodforoddq}A assigns each pair of the given inputs to at least one reducer in common.
\end{theorem}
\begin{proof}
We have $(y-1)\big\lceil\frac{y}{2}\big\rceil$ pairs of inputs of the set $A$ of size $q-1$, and there are the same number of reducers; hence, each reducer can hold one input pair. Further, the remaining capacity of all the reducers of each team can be used to assign an input of $B$. Hence, all the inputs of $A$ are paired with every other input and every input of $B$ (as we proved in Claims~\ref{claim:algo_odd_q_claim_2} and~\ref{claim:algo_odd_q_claim_3}). Following the fact that the inputs of the set $A$ are paired with all the $m$ inputs, the inputs of the set $B$ is also paired by following a similar procedure on them. Thus, Algorithm~\ref{alg:2stepmethodforoddq}A assigns each pair of the given $m$ inputs to at least one reducer in common.
\end{proof}

\medskip
\section{Pseudocode and Correctness of Algorithm 1B}
\label{app:algo_correctness_2}

\LinesNotNumbered
\setcounter{algocf}{0}
\begin{algorithm}[h]
\textbf{Inputs:} $m$: the number of bins obtained after placing all the given $m^{\prime}$ inputs (of size $\leq \frac{q}{k}$, $k\geq 4$ is an even number) to bins each of size $\frac{q}{k}$,

$q$: the reducer capacity.

{\bf Variables:}

$\mathit{Team}[i,j]:$ represents teams of reducers, where index $i$ indicates $i^{th}$ team and index $j$ indicates $j^{th}$ reducer in $i^{th}$ team. Consider $u = \frac{2m}{q}$. There are $u-1$ teams of $\big\lceil \frac{u}{2}\big\rceil$ reducers in each team.

$groupA[]:$ represents disjoint groups of inputs of the set $A$, where $groupA[i]$ indicates $i^{th}$ group of $\big\lceil\frac{q}{2}\big\rceil$ inputs of the set $A$.

\nl{\bf Function $create\_group(m)$} \nllabel{ln:function_create_group_even}
\Begin{

\nl \lFor{$i\leftarrow 1$ \KwTo $u$}{$groupA[i]\leftarrow \langle i, i+1\ldots, i+\frac{q}{2}-1\rangle, i\leftarrow i+\frac{q}{2}$ \nllabel{ln:make_group_even1}}

\nl $2\_step\_even\_q(1, u)$, $\mathit{Last\_Team(1, \big\lceil \frac{u-1}{2}\big\rceil, u)}$ \nllabel{ln:call_43-method_even}}

\nl{\bf Function $2\_step\_even\_q(lower, upper)$} \nllabel{ln:function_2step_method_even}
\Begin{

\nl \lIf{$\big\lfloor \frac{upper-lower}{2}\big\rfloor < 1$}{\KwRet}

\nl \Else{

\nl $mid\leftarrow \big\lceil\frac{upper-lower}{2}\big\rceil$, $\mathit{Assignment(lower, mid, upper)}$ \nllabel{ln:call_recursion_even}

\nl $2\_step\_even\_q(lower, mid)$, $2\_step\_even\_q(mid+1, upper)$ \nllabel{ln:first_second_divide_even}}}

\caption{{\textbf{Part B}}}
\label{alg:2stepmethodforevenq}
\end{algorithm}

We show that every pair of inputs is assigned to reducers. Specifically, Algorithm~\ref{alg:2stepmethodforoddq}B satisfies two claims, as follows:

\medskip\begin{claim}\label{claim:algo_even_q_claim_1}
Pairs of derived inputs $\langle i, j\rangle$ and $\langle i^{\prime}, j^{\prime}\rangle$, where $i\neq i^{\prime}$ or $j\neq j^{\prime}$, are assigned to a team.
\end{claim}

\medskip\begin{claim}\label{claim:algo_even_q_claim_2}
All the given $m$ inputs appear exactly once in each team.
\end{claim}
We do not prove Claims~\ref{claim:algo_even_q_claim_1} and~\ref{claim:algo_even_q_claim_2}. Note that Claim~\ref{claim:algo_even_q_claim_1} follows Claims~\ref{claim:algo_odd_q_claim_1}, where Claims~\ref{claim:algo_odd_q_claim_1} shows that all the pairs of inputs of the set $A$ (in case $q=3$) and all the pairs of derived inputs of the set $A$ (in case $q>3$) $\langle i,j\rangle$ and $\langle i^{\prime},j^{\prime}\rangle$, where $i\neq i^{\prime}$ or $j\neq j^{\prime}$ are assigned to a team. Also, Claim~\ref{claim:algo_even_q_claim_2} follows Claim~\ref{claim:algo_odd_q_claim_2}, where Claim~\ref{claim:algo_odd_q_claim_2} shows that all the inputs of the set $A$ appear in each team only once, while in case of Algorithm~\ref{alg:2stepmethodforoddq}B the set $A$ is considered as a set of $m$ inputs.

\medskip\begin{theorem}\label{th:algo_even_q_theorem}
Algorithm~\ref{alg:2stepmethodforoddq}B assigns each pair of the given inputs to at least one reducer in common.
\end{theorem}
\begin{proof}
Since there are the same number of pairs of the derived inputs and the number of reducers, it is possible to assign one pair to each reducer that results in all-pairs of the $m$ inputs.
\end{proof}

\medskip
\section{Correctness of Algorithm 2}
\label{app:algo_correctness_3}
The correctness shows that all-pairs of inputs are assigned to reducers. Specifically, we show that each pair of inputs of the set $A$ is assigned to $p(p+1)$ reducers that use only $p$ capacity of each reducer (Claims~\ref{claim:first_ext_AU_claim_1} and~\ref{claim:first_ext_AU_claim_2}). Then, we prove that the set $B$ holds $x\leq m-p^2$ inputs. At last we conclude that Algorithm 2 assigns each pair of inputs to reducers.

\medskip\begin{claim}\label{claim:first_ext_AU_claim_1}
All the inputs of the set $A$ are assigned to $p(p+1)$ reducers, and the assignment of the inputs of the set $A$ uses only $p$ capacity of each reducer.
\end{claim}

\medskip\begin{claim}\label{claim:first_ext_AU_claim_2}
All the inputs of the set $A$ appear in each team exactly once.
\end{claim}
We are not proving Claims ~\ref{claim:first_ext_AU_claim_1} and~\ref{claim:first_ext_AU_claim_2} here. Claims~\ref{claim:first_ext_AU_claim_1} and~\ref{claim:first_ext_AU_claim_2} follow the correctness of the \emph{AU method}; hence, all the inputs of the set $A$ are placed to $p+1$ teams of $p$ bins (each of size $q$) in each team, and the assignment of each such bin only uses $p$ capacity of each reducer. Further two bins cannot be assigned to a reducer because $2\times p>q$. Claim~\ref{claim:first_ext_AU_claim_2} also follows the correctness of the \emph{AU method}, and hence, all the inputs of the set $A$ appear only once in each team.

\medskip\begin{claim}\label{claim:first_ext_AU_claim_3}
When the reducer capacity is $q$, the set $B$ holds $x\leq m-p^2$ inputs, where $p$ is the nearest prime number to $q$.
\end{claim}
\begin{proof}
There are $p+1$ teams of $p$ reducers in each team, and inputs of the set $A$ use $q-p$ capacity of each of the reducers. Hence, each reducer can hold $q-p$ additional unit-sized (almost identical-sized) inputs. Since inputs of the set $A$ appear in each team (Claim~\ref{claim:first_ext_AU_claim_2}), an assignment of $q-p$ additional unit-sized inputs to all the reducers of a team provides pairs of all the inputs of the set $A$ with additional inputs. In this manner, $p+1$ teams, which hold $p^2$ inputs of the set $A$, can hold at most $(p+1)\times(q-p)$ additional inputs. Since $p^2<m\leq p^2+(p+1)\times(q-p)$, the set $B$ can hold $x\leq m-p^2$ inputs.
\end{proof}

\medskip\begin{theorem}\label{th:first_ext_AU_theorem}
Algorithm 2 assigns each pair of inputs to reducers.
\end{theorem}
We are not proving Theorem~\ref{th:first_ext_AU_theorem} here. The proof of Theorem~\ref{th:first_ext_AU_theorem} considers the fact that all the inputs of the set $A$ are paired with each other using the \emph{AU method}, and they are also paired with all the remaining inputs of the set $B$. Further, inputs of the set $B$ will be paired with each other by using Algorithm~\ref{alg:2stepmethodforoddq}A or~\ref{alg:2stepmethodforoddq}B (Theorems~\ref{th:algo_odd_q_theorem} or~\ref{th:algo_even_q_theorem}).


\end{document}